\documentclass[journal]{IEEEtran}
\usepackage{amsmath,amsfonts,amssymb}
\usepackage{algorithmic}
\usepackage{array}
\usepackage[caption=false,font=normalsize,labelfont=sf,textfont=sf]{subfig}
\usepackage{textcomp}
\usepackage{stfloats}
\usepackage{url}
\usepackage{verbatim}
\usepackage{graphicx}
\usepackage{orcidlink}
\usepackage{diagbox}
\def\BibTeX{{\rm B\kern-.05em{\sc i\kern-.025em b}\kern-.08em
    T\kern-.1667em\lower.7ex\hbox{E}\kern-.125emX}}

\usepackage{subfig}
\usetikzlibrary{positioning, arrows.meta, shapes.geometric, calc}
\usetikzlibrary{patterns}

\makeatletter
\def\@begintheorem#1#2{\trivlist\item[\hskip\labelsep{\bfseries #1\ #2.}]\itshape}
\def\@opargbegintheorem#1#2#3{\trivlist\item[\hskip\labelsep{\bfseries #1\ #2\ (#3).}]\itshape}
\makeatother
\newtheorem{theorem}{Theorem}
\newtheorem{lemma}{Lemma}

\definecolor{cRS}{RGB}{60,120,220}      
\definecolor{cFB}{RGB}{230,140,40}      
\definecolor{cDATA}{RGB}{50,170,120}    
\definecolor{cINT}{RGB}{220,70,70}      
\definecolor{cBOX}{RGB}{245,248,255}

\tikzset{
  pics/basestation/.style={
    code={
      \fill[black!10] (-0.25,0) -- (0,1.3) -- (0.25,0) -- cycle;
      \draw[black, line width=0.9pt] (-0.25,0) -- (0,1.3) -- (0.25,0);
      \draw[black, line width=0.9pt] (-0.18,0.35) -- (0.18,0.35);
      \draw[black, line width=0.9pt] (-0.13,0.70) -- (0.13,0.70);
      \draw[black, line width=0.9pt] (-0.08,1.05) -- (0.08,1.05);
      \draw[black, line width=0.9pt] (0,1.3) -- (0,1.55);
      \draw[black, line width=0.8pt] (0,1.55) ++(0.10,0.05)
        arc[start angle=20,end angle=160,radius=0.18];
      \draw[black, line width=0.8pt] (0,1.55) ++(0.14,0.10)
        arc[start angle=20,end angle=160,radius=0.28];
      \draw[black, line width=0.8pt] (0,1.55) ++(0.18,0.15)
        arc[start angle=20,end angle=160,radius=0.38];
    }
  },
  pics/phone/.style={
    code={
      \draw[black, line width=0.9pt, rounded corners=5pt, fill=black!6]
        (-0.55,-1.05) rectangle (0.55,1.05);
      \draw[black!60, line width=0.6pt, rounded corners=3pt, fill=white]
        (-0.45,-0.70) rectangle (0.45,0.70);
      \draw[black!70, line width=0.6pt] (-0.12,0.86) -- (0.12,0.86);
      \fill[black!55] (0,-0.86) circle (0.05);
    }
  }
}

\tikzset{
  arrRS/.style={-{Latex[length=3mm]}, line width=1.0pt, draw=cRS, dashed},
  arrFB/.style={-{Latex[length=3mm]}, line width=1.0pt, draw=cFB},
  arrDATA/.style={-{Latex[length=3mm]}, line width=1.1pt, draw=cDATA},
  arrINT/.style={-{Latex[length=3mm]}, line width=1.0pt, draw=cINT, dashed},
  callout/.style={rounded corners=6pt, draw=black!50, fill=cBOX, line width=0.7pt,
                  inner sep=6pt, align=left},
  label/.style={font=\small}
}

\begin{document}
\bstctlcite{BSTcontrol}
\title{Rate Loss Analysis for Multiple-Antenna NOMA with Limited Feedback}
\author{Ruizhan Shen\orcidlink{0000-0001-6591-6333} and Hamid Jafarkhani\orcidlink{0000-0001-6838-8038}, \IEEEmembership{Fellow, IEEE} \thanks{The authors are with the University of California Irvine, Irvine, CA 92697 USA. This work was supported in part by the NSF Award CCF-2328075. }
\thanks{A shorter version of this work has been submitted to IEEE for possible publication.
Copyright may be transferred without notice, after which this version may no longer be accessible.}}

\maketitle

\begin{abstract}
In the limited feedback downlink multiple-input single-output (MISO) non-orthogonal multiple access (NOMA) system, both the effective channel gain and the channel direction need to be quantized. The quantization error affects the feasible region of NOMA and the rate loss compared with the case of full channel state information (CSI). In this work, we analyze these effects and obtain an upper bound for the rate loss.
Numerical results show that the sum rate of the limited feedback MISO-NOMA system approaches that of the full CSI as the number of feedback bits increases.
\end{abstract}


\section{Introduction}
\IEEEPARstart{A}{s} a promising solution for the next-generation multiple access, non-orthogonal multiple access (NOMA) has been widely studied~\cite{10720669,10786246}. 
A key advantage of NOMA is the ability to support more devices.
By allowing users to share the same time, frequency, and code resources, NOMA can increase the overall efficiency of wireless networks compared with orthogonal multiple access (OMA), especially user fairness and quality of service (QoS) performance~\cite{Norouzi2023Joint}~\cite{8823873}. The choice of multiple access depends on the objective and constraints, while rate-splitting multiple access (RSMA) provides a general framework encompassing spatial-domain multiple access (SDMA), OMA, and NOMA~\cite{9832611}.

In this work, we focus on downlink power-domain NOMA. More specifically, we consider beamformer-based multiple-input multiple-output (MIMO)-NOMA systems~\cite{10729214}. Similar to SDMA, beamformer-based MIMO-NOMA uses precoding but in contrast, it applies successive interference cancellation (SIC) decoding.
The beamformer-based MIMO-NOMA solutions in the literature are generally based on full-CSI, which is not available in frequency-division duplex systems~\cite{10729214}. 

Although limited feedback for SDMA has been well studied and adopted by standards~\cite{3gpp_ts38214_v16.17}, research on limited feedback for MIMO-NOMA is not as mature, especially for beamformer-based systems. To use NOMA with limited feedback, Ding et al.~\cite{7434594} proposed considering 1-bit feedback that indicates whether the effective gain of each user equipment (UE) exceeds a preset threshold. Cui et al.~\cite{8502922} proposed a successive convex approximation algorithm to design beamforming and power allocation with channel distribution information or channel estimation uncertainty. Yapici et al.~\cite{9094017} proposed a two-bit mmWave NOMA framework, where the two bits indicate whether the user distance and channel angle exceed predefined thresholds. Moreover, Morales-C\'espedes et al.~\cite{8935164} present a NOMA power allocation method in which only large-scale effects are available to the transmitter. Although No-CSI NOMA designs have also been developed for outage performance using reinforcement-learning techniques~\cite{9281361}, CSI-assisted designs are more suitable for instantaneous adaptation as they do not require the online learning phase and can support the optimization of broader design objectives.

The effects of limited feedback in system performance have been studied for different SISO-NOMA systems, including asynchronous and reconfigurable intelligent surface-aided systems ~\cite{7968348,9103094,10384715}. However, the limited feedback design and analysis developed for SISO-NOMA systems cannot be applied directly to MISO-NOMA systems. In SISO-NOMA, users can be ordered according to their scalar channel strength, and the limited feedback focuses mainly on quantizing channel gains.
In contrast, in multiple antenna systems, the base station (BS) must also select the beamforming vector, which directly affects the users' effective channel gain and the optimal decoding order. 
The impacts of the feedback rate on the probability of outage and the optimal number of feedback bits to maximize net throughput in cluster-based NOMA have been studied in \cite{7972957}. However, to the best of our knowledge, the effects of limited feedback on beamforming and the optimal decoding order have not been considered in the existing literature. In addition, the analysis of limited feedback beamformer-based NOMA does not exist in the literature.


In this work, we consider a two-user MISO-NOMA system, which is a common practical setting and can serve as the basic cluster structure in multi-cluster NOMA, thereby providing insights into more general multi-user scenarios. The main contributions of this work are as follows:
\begin{itemize}
    \item We present a limited feedback multiple-antenna NOMA framework which considers the interaction of beamforming, power allocation, and SIC decoding order. 
    \item We analyze the rate loss caused by limited feedback, provide an upper bound, and characterize the convergence rate for Rayleigh channels. 
    \item We analyze and present the feasible region where one should use NOMA with limited feedback.
\end{itemize}
The work is organized as follows: Section \ref{sec:model} includes the system model, problem formulation, and the quantizer structure. Section \ref{sec:region} presents the feasible regions. Section \ref{sec:loss} provides a rate loss analysis, where we obtain an upper bound for the rate loss. Simulation results are presented in Section \ref{sec:simu}. Section \ref{conclusion} concludes the work. 

\textbf{\textit{Notations:}} In this work, regular and bold letters stand for scalars and vectors, respectively. $(\cdot)^H$ represents the Hermitian transpose. In addition, $|x|$ and $\mathbb{E}[x]$ denote the absolute and expected values of x, respectively. $\|\mathbf{x}\|$ represents the Euclidean norm of vector $\mathbf{x}$ and $x \propto y$ means that $x$ is linearly proportional to $y$.

\section{System Model}\label{sec:model}
\subsection{Transmit and Feedback Process}
We consider a limited feedback beamformer-based system, as shown in Fig. \ref{fig:1}. The BS has $N_t$ antennas and transmits the following superposed signal to two single-antenna users:
\[
\mathbf{x} = \mathbf{w}_1 \sqrt{p_1} s_1 + \mathbf{w}_2 \sqrt{p_2} s_2, 
\]
where $\mathbf{w}_k \in \mathbb{C}^{N_t\times1}$ is a unit-norm precoding vector chosen from a codebook $\mathcal{C}=\{ \mathbf{c}_i \in \mathbb{C}^{N_t\times1}: ||\mathbf{c_i}||=1,i=1,2,\dots 2^{B'}\}$, $B'$ is the number of feedback bits that represent the precoding vector indices. The predefined codebook is shared between BS and UEs. The channel vector and power allocation of User $k$ are $\mathbf{h}_k \in \mathbb{C}^{N_t\times1}$ and $p_k$, respectively. The noise power is $\sigma^2$.

\begin{figure}[ht]
\centering
\includegraphics[width=0.92\linewidth]{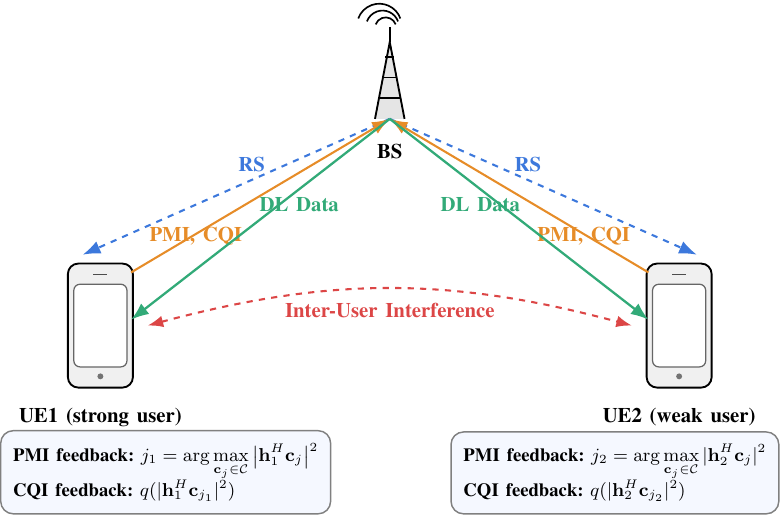}
\caption{Beamformer-based NOMA signaling with limited feedback.}
\label{fig:1}
\end{figure}

Each UE selects a precoding vector using its local CSI based on maximum ratio transmission (MRT), 
$
j_k=\arg\max_{\mathbf{c}_j\in\mathcal{C}} \big|\mathbf{h}_k^{H}\mathbf{c}_j\big|^2.
$
User $k$ sends the index $j_k$ with $B'$ bits and the quantized effective gain $q(|\mathbf{h}_k^{H}\mathbf{c}_{j_k}|^2)$, as channel quality indicator (CQI), with $B$ bits to BS. The quantizer is discussed in Section~\ref{sec:quan}. Denoting $\mathbf{w}_k=\mathbf{c}_{j_k}$ as the beamforming vector used by BS for User k, we call the user with the higher effective channel gain $|\mathbf{h}_k^H \mathbf{w}_k|^2$, the strong user. Using the CQI feedback, the BS can decide which user is the strong user and share this information with 1 bit.

\subsection{Sum Rate Maximization}
Without loss of generality, we assume that User 1 is the strong user and $p_1 = \beta P$ and $p_2 = (1-\beta)P$, where P is the total power. 
User 1 first decodes User 2's message with the following signal to interference plus noise ratio (SINR):
\[
    SINR_{1\to2}=\tfrac{(1-\beta)P |\mathbf{h_1^H}\mathbf{w_2}|^2}{ \beta P |\mathbf{h_1^H}\mathbf{w_1}|^2+\sigma^2},
\]
After that Users 1 and 2 decode their own messages with
$
    SINR_{1}=\tfrac{\beta P |\mathbf{h_1^H}\mathbf{w_1}|^2}{\sigma^2},\ \ \ \ \ \ \ \
    SINR_{2}=\tfrac{(1-\beta) P |\mathbf{h_2^H}\mathbf{w_2}|^2}{ \beta P |\mathbf{h_2^H}\mathbf{w_1}|^2+\sigma^2}.
$
\newline
We consider maximizing the sum rate with QoS constraints:
\[
\begin{aligned}
\max_{\beta}\quad & R_1 + R_2 \\
\text{s.t.}\quad 
 & R_1 \ge R_{\mathrm{th}}, R_2 \ge R_{\mathrm{th}},\\
 & \mathbf w_1,\mathbf w_2 \in \mathcal C .
\end{aligned}
\]
This optimization concerns the achievable sum rate of beamformer-based NOMA with linear precoding and SIC, not the capacity region of the general MIMO broadcast channel.

Successful decoding requires that the strong user decodes its own message and the weak user's message, and the weak user decodes its own message. Let $\varepsilon=2^{R_{\mathrm{th}}}-1$, then successful decoding requires $\mathrm{SINR}_2(\beta)$, $\mathrm{SINR}_{1\to2}(\beta)$, $\mathrm{SINR}_{1}(\beta) \geq \varepsilon$, resulting in the following separate conditions, respectively:
$
\beta \ge \beta_{1,\min}=\frac{\varepsilon\,\sigma^2}{P\,|\mathbf{h}_1^H\mathbf{w}_1|^2},
\beta \le \beta_{2,\max}=\frac{P\,|\mathbf{h}_2^H\mathbf{w}_2|^2-\varepsilon\sigma^2}{P\,|\mathbf{h}_2^H\mathbf{w}_2|^2+\varepsilon\,P\,|\mathbf{h}_2^H\mathbf{w}_1|^2},
\beta \le \beta_{\mathrm{SIC},\max}=\frac{P\,|\mathbf{h}_1^H\mathbf{w}_2|^2-\varepsilon\sigma^2}{P\,|\mathbf{h}_1^H\mathbf{w}_2|^2+\varepsilon\,P\,|\mathbf{h}_1^H\mathbf{w}_1|^2}
$.
Although $\beta$ can take any value in $[0,1]$  when there is no rate constraint, the $\beta$ values within the above range satisfy the QoS and SIC rate constraints.
Here, we assume that the BS directly uses the precoding matrix indicator (PMI) recommended by the UEs. In the full-CSI case, we adopt MRT beamforming with beamforming vectors aligned with the corresponding channel vectors, and therefore we have $
\beta^*_{1,\min}=\frac{\varepsilon\,\sigma^2}{P\,\|\mathbf h_1\|^2},
\beta^*_{2,\max}=\frac{P\,\|\mathbf h_2\|^2-\varepsilon\sigma^2}{P\,\|\mathbf h_2\|^2+\varepsilon\,P\,\left|\mathbf{h}_2^H\mathbf{h}_1/\|\mathbf{h}_1\|\right|^2},
\beta^*_{\mathrm{SIC},\max}=\frac{P\,\left|\mathbf{h}_1^H\mathbf{h}_2/\|\mathbf{h}_2\|\right|^2-\varepsilon\sigma^2}{P\,\left|\mathbf{h}_1^H\mathbf{h}_2/\|\mathbf{h}_2\|\right|^2+\varepsilon\,P\,\|\mathbf h_1\|^2}
$.
The optimal full-CSI power allocation maximizes the sum rate, i.e., $\beta^*=\arg\max_{\beta\in\mathcal{B}} \mathrm{R}(\beta)$, where 
$\mathcal{B} = \{\beta| \beta\ge\max\{0, \beta^*_{1,min}\}, \beta\le \min\{\beta^*_{2,\max},\ \beta^*_{\mathrm{SIC},\max}, 1\}$. As discussed later, if $\mathcal{B}$ is empty, NOMA is not feasible.

\subsection{Power Allocation with Limited Feedback}
Power allocation requires both effective and interference gains, but only the former is directly available through feedback. BS can use quantized gains and beamforming direction to estimate the real channel, $\hat{\mathbf{h}}_i \approx \sqrt{\hat H_{ii}}\mathbf{w}_i$, where
$\hat H_{ii}=q(|\mathbf{h}_i\mathbf{w}_i|^2)$ is the quantized effective gain. The interference gain can be estimated as
$\hat H_{ij}
    = |\hat h_i^H \mathbf{w}_j|^2
     = \hat H_{ii}\,|\mathbf{w}_i^H \mathbf{w}_j|^2.$
Then, the BS can determine the optimal power allocation, 
$\beta_q=\arg\max_{\beta\in\mathcal{B}_q} \mathrm{R}(\beta)$, where 
$\mathcal{B}_q = \{\beta| \beta\ge\max\{0, \hat\beta_{1,min}\}, \beta\le \min\{\hat\beta_{2,\max},\ \hat\beta_{\mathrm{SIC},\max}, 1\}$. The parameters are reevaluated with limited feedback, i.e., $\hat\beta_{1,\min}=\frac{\varepsilon\,\sigma^2}{P\,\hat H_{11}}$, $
\hat\beta_{2,\max}=\frac{P\,\hat H_{22}-\varepsilon\,\sigma^2}{P\,\hat H_{22}+\varepsilon\,P\,\hat{H}_{21}}$, and $
\hat\beta_{\mathrm{SIC},\max} =\frac{P\,\hat{H}_{12}-\varepsilon\,\sigma^2}{P\,\hat{H}_{12}+\varepsilon\,P\,\hat H_{11}}.$

\subsection{Magnitude and Direction Quantization}\label{sec:quan}

There are two components in limited feedback systems that cause performance degradation: (i) direction quantization, which limits the beamforming choices, and (ii) magnitude quantization, which affects SIC order and power allocation.

To describe the effect of direction quantization, let us denote $H_{i}=\|\mathbf h_i\|^2$ as the effective gain with ideal MRT, and $H_{ii}=|\mathbf h_i^{H}\mathbf w_i|^2$ as the effective channel gains after beamforming with limited feedback. Then, $\eta_{ii} = \frac{H_{ii}}{H_i}$ represents the effective gain loss due to the non-ideal beamforming direction. Similarly, $H_{ij}=|\mathbf h_i^{H}\mathbf w_j|^2, i\neq j$ represents the gain of cross interference and $\eta_{ij} = \frac{H_{ij}}{H_i}, i\neq j$ defines the ratio of interference to effective gain. For the purpose of analysis, we assume that the beamforming codebook is a random vector quantization (RVQ) codebook.  
To guarantee that the BS does not transmit at a rate more than the capacity of the weak user's channel, we borrowed the quantizer from~\cite{7968348} which quantizes the input values to their left boundries:
\[
q(x) = 
\begin{cases}
    \left\lfloor \dfrac{x}{\delta} \right\rfloor \delta, & x < (2^B-1)\delta\\
    (2^B-1)\delta, &x\ge (2^B-1)\delta
\end{cases}
\]

When $x < (2^B-1)\delta$, we have $|H_{ii}-\hat H_{ii}|<\delta$, and $\delta\propto 2^{-B}$. When $H_{ii}>(2^B-1)\delta$, the quantization error, $|H_{ii}-(2^B-1)\delta|$, is denoted by $\Delta_{sat}$.
The magnitude quantization has a constant computational complexity, while the direction quantization requires searching $2^{B'}$ code vectors using $N_t$-dimension inner products, resulting in $\mathcal{O}(2^{B'N_t})$ complexity.

\section{Feasible Region}\label{sec:region}

Under certain conditions, NOMA cannot guarantee (i) the minimum rate requirements or (ii) SIC decoding. In such a region, NOMA is not feasible.
In this section, we present the sufficient channel conditions for the feasibility of NOMA.

\subsection{Feasible Region in Full-CSI}
In full-CSI systems, meeting the minimum rate requirements and ensuring successful SIC decoding is equivalent to the following properties of $\beta^*$:

\begin{itemize}

\item $\beta^*$ exists: $\beta_{1,\min} \le \beta_{2,\max}$ and $\beta_{1,\min} \le \beta_{\mathrm{SIC},\max}$.

\item $\beta^*$ is within $(0,1)$: 
$\beta_{\mathrm{1},\min} \le 1$, i.e., $P\|\mathbf{h}_1\|^2\ge \epsilon\sigma^2$.
\end{itemize}

In general, the feasibility of NOMA can be decided by the disparity of the channel strengths, $\rho =\frac{\|\mathbf{h}_2\|}{\|\mathbf{h}_1\|}$, and the angle difference, captured by $\cos\theta = \frac{|\mathbf{h}_1^H\mathbf{h}_2|}{\|\mathbf{h}_1\|\|\mathbf{h}_2\|}$~\cite{9832611}. Simple algebra shows that the existence of $\beta^*$ is equivalent to: $\rho > 1/\sqrt{\tfrac{P\|\mathbf{h}_1\|^2}{\epsilon\sigma^2}-1-\epsilon \cos^2\theta}$ and $\cos\theta>\sqrt{\tfrac{\epsilon\sigma^{2}(\epsilon+1)}
           {P\|{\bf h}_1\|^{2}-\epsilon\sigma^{2}}}$.
When these conditions hold, NOMA is feasible and we use it. Otherwise, we use SDMA or OMA instead.

\subsection{Feasibility Decision in Limited-Feedback}
In systems with limited feedback, the BS can decide whether to use NOMA based on similar conditions:
\begin{itemize}
\item $\beta_q$ exists.
\item $\beta_q$ is within $(0,1)$: $P\hat H_{11}\ge \epsilon\sigma^2$.
\end{itemize}
The existence of $\beta_q$ results in  $\hat\rho > 1/\sqrt{\tfrac{P\hat{H}_{11}}{\epsilon\sigma^2}-1-\epsilon \cos^2\hat\theta}$ and $\cos\hat\theta>\sqrt{\tfrac{\epsilon\sigma^{2}(\epsilon+1)}
           {P \hat{H}_{11}-\epsilon\sigma^{2}}}$, where $\cos\hat{\theta}=\tfrac{|\mathbf{w}_1^H\mathbf{w}_2|}{\|\mathbf{w}_1\|\|\mathbf{w}_2\|}$ and $\hat{\rho}=\sqrt{\tfrac{\hat{H}_{22}}{\hat{H}_{11}}}$. 

However, due to quantization errors, applying the same decision criterion as in the full CSI case leads to marginal regions where the feasibility of NOMA is uncertain. Because of quantization, $\hat{H}_{11} \le \eta_{11}\|\mathbf{h}_1\|^2 < \hat{H}_{11}+\delta$ and a sufficient condition for the feasibility of NOMA is  $\cos\hat{\theta}>\sqrt{\tfrac{\epsilon\sigma^{2}(\epsilon+1)} {P (\eta_{11}H_1-\delta)-\epsilon\sigma^{2}}}$ and $\hat{\rho} > 1/\sqrt{\tfrac{P(\eta_{11}{H}_{1}-\delta)}{\epsilon\sigma^2}-1-\epsilon \cos^2\hat{\theta}}$. The areas above the blue curve in Fig.~\ref{fig:3} show the feasible region where NOMA is feasible in the $(\cos\theta,\rho)$ plane.

Note that $\hat{H}_{11}=0$ and $\hat{H}_{22}=0$ result in $\hat{\beta}_{1,min}=\infty$ and $\hat{\beta}_{2,max}<0$, respectively. In both cases, NOMA is not feasible, and SDMA or OMA is adopted instead. This sufficient condition is analytical. In operation, the BS uses the reported PMI and CQI to determine if NOMA is feasible by checking whether $\mathcal{B}_q$ is nonempty.

\begin{figure}[ht]
\centering
\includegraphics[width=\linewidth]{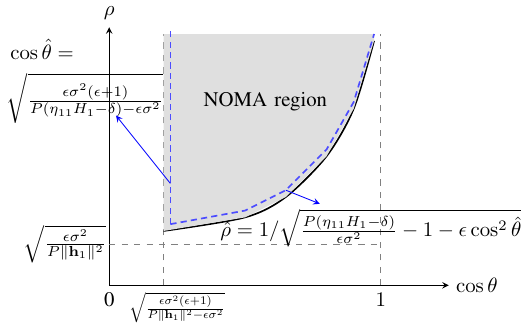}
\caption{Full-CSI and limited feedback NOMA feasible regions in the $(\cos\theta,\rho)$ plane.}
\label{fig:3}
\end{figure}

\section{Rate Loss Analysis}\label{sec:loss}
When NOMA is the selected scheme and is feasible in both full-CSI and limited feedback, we can evaluate the rate loss compared with the full CSI case by
\[
\Delta R(H,\hat H)=R_{\mathrm N}^{\mathrm{full}}(H;\beta^*)-R_{\mathrm N}^{\mathrm{LF}}(H;\beta_q).
\]
When NOMA is feasible, there are two cases: (i) $\delta<H_{ii}<(2^B-1)\delta$ and (ii) $H_{ii}>(2^B-1)\delta$. In what follows, we will derive the results for Case (i). The results for Case (ii) can be obtained similarly by replacing $\delta$ with $\Delta_{sat}$.

\begin{theorem}\label{thm:1}
The achievable rate loss is upper bounded by 
\begin{equation}
\begin{aligned}
   \Delta R
& \le C_1H_1(1-\eta_{11})+C_2H_2^2\sqrt{2-2\eta_{11}}\\&+C_3H_2(1-\eta_{22})
+C_4(H_1+H_2)|\Delta \beta|),
\end{aligned}
\end{equation}
where $C_1$ to $C_4$ are constants that are independent of $B$ and $B'$, $\Delta\beta\triangleq\beta_q-\beta^*$ is the difference in power allocation.
\end{theorem}

\noindent\textit{Proof}. See Appendix ~\ref{app:thm1}.

\begin{lemma}\label{le:1}
In the above MISO-NOMA system, the optimal $\beta^*$ is either $\beta_{\min}=\beta_{1,\min}$, $\beta_{\max}=\min\{\beta_{2,\max}, \beta_{SIC,\max}\}$, or the stationary point $\beta_0$.
\end{lemma}

\noindent\textit{Proof}. See Appendix ~\ref{app:lemma1}.

\begin{lemma}\label{le:2}
For Rayleigh channels, the loss induced by $\left|\Delta\beta\right|$ converges to 0 at least exponentially as $B$ and $B'$ increase.
\end{lemma}
\noindent\textit{Proof}. See Appendix ~\ref{app:lemma2}.

\begin{theorem}\label{thm:2}
    For Rayleigh channels, the achievable rate loss converges to zero at least exponentially as $B$ and $B'$ increase.
\end{theorem}
\noindent\textit{Proof}. Since $\mathbf{h}\sim \mathcal{CN}(\mathbf{0},\mathbf{I_{N_t}})$, $H_i$ follows the Gamma distribution, $\mathbb{E}[H_1] = N_t$. For the Rayleigh channel and RVQ, the CDF of $\eta_{ii}$ is $F_{\eta}(\eta)=\left(1-(1-\eta)^{N_t-1}\right)^{2^{B'}}$~\cite{4100151}. Therefore, its expectation is~\cite{1715541}: $\mathbb{E}[\eta_{ii}] \approx 1-2^{-B'/(N_t-1)}$.

Since $\eta_{11}$ and $\eta_{22}$ depend only on the direction difference, they are independent of the magnitude terms. It can be proved that each item converges to 0 at least exponentially as follows:
\begin{equation}
\begin{aligned}
    \mathbb{E}[C_1H_1(1-\eta_{11})]&=C_1\mathbb{E}[H_1]\mathbb{E}[1-\eta_{11}]\\
    &\approx C_1N_t 2^{-B'/(N_t-1)},
\end{aligned}
\end{equation}

\begin{equation}
\begin{aligned}
    \mathbb{E}[C_2H_2^2\sqrt{2-\eta_{11}}]&=C_2 \mathbb{E}[H_2^2]\mathbb{E}[\sqrt{2-\eta_{11}}]\\
    &\le C_2\mathbb{E}[H_2^2]\sqrt{\mathbb{E}(1-\eta_{11})}\\
    &\approx N_t(N_t+1) 2^{-B'/2(N_t-1)},
\end{aligned}
\end{equation}

\begin{equation}
    \mathbb{E}[C_3H_2(1-\eta_{22})]\approx C_3N_t 2^{-B'/(N_t-1)}.
\end{equation}

For the fourth item, from Cauchy–Schwarz inequality,
\begin{equation}
    \mathbb{E}[(H_1+H_2)|\Delta\beta|]\le \sqrt{\mathbb{E}(H_1+H_2)^2 \mathbb{E}(|\Delta\beta|)^2}.
\end{equation}
Since we assume that $\mathbf{h}_1$ and $\mathbf{h}_2$ are independent,
$\mathbb{E}(H_1+H_2)^2 = \mathbb{E} H_1^2 +\mathbb{E} H_2^2 + 2\mathbb{E}[H_1H_2]= 2N_t(N_t+1)+2N_t^2$. Therefore, we have 
\begin{equation}
    \mathbb{E}[(H_1+H_2)|\Delta\beta|]\le \sqrt{(2N_t(N_t+1)+2N_t^2) \mathbb{E}(|\Delta\beta|)^2}.
\end{equation}
Also, since $|\Delta\beta|<1$, we know $\mathbb{E}|\Delta\beta|^2\le \mathbb{E}|\Delta\beta|$. From Lemma~\ref{le:2}, $\mathbb{E}|\Delta\beta|$ converges to 0 at least exponentially, and we have the same conclusion for the fourth item. This concludes the proof of Theorem~\ref{thm:2}.

\section{Simulation Results}\label{sec:simu}
To show the effects of limited feedback, we simulate a two-transmit-antenna system with SNR 10 dB, rate threshold $1\ bit/s/Hz$ and $10^6$ samples. 
The value of the quantizer step size $\delta$, for a given number of bits, is optimized offline by a 1-D numerical search minimizing the quantization MSE using $10^5$ samples from $|\mathbf{h^H \mathbf{w}}|^2$, where the distribution of $\mathbf{h}$ is Rayleigh. Table~\ref{tab:table2} lists the optimized $\delta$ for different number of feedback bits.
In the simulation, NOMA feasibility is assessed separately at the BS using CSI available in each case, while rates are evaluated using the true channel.

\begin{table}[ht]
\centering
\caption{Value of $\delta$ for different $B$ and $B'$, trained with $|\mathbf{h^H \mathbf{w}}|^2$}
\label{tab:table2}
\begin{tabular}{l|cccccc}
\diagbox{$B'$}{$B$} & 1 & 2 & 3 & 4 & 5 & 6\\
\hline
1 & 1.59 & 0.97 & 0.60 & 0.36 & 0.22 & 0.13\\
2 & 1.70 & 1.06 & 0.65 & 0.39 & 0.23 & 0.13\\
3 & 1.81 & 1.11 & 0.69 & 0.42 & 0.24 & 0.14\\
4 & 1.92 & 1.16 & 0.71 & 0.43 & 0.25 & 0.14\\
5 & 1.98 & 1.20 & 0.73 & 0.44 & 0.26 & 0.15\\
6 & 2.00 & 1.22 & 0.75 & 0.44 & 0.26 & 0.15\\
\hline
\end{tabular}
\end{table}

Fig. \ref{fig:4} shows the average sum rate loss versus $B$ when $B'$ is fixed, considering the case where NOMA is feasible in both full CSI and limited feedback scenarios. The rate loss hardly improves beyond $B=3$ for fixed $B'$, and beyond $B'=5$ for fixed $B$.

Fig. \ref{fig:5} shows the average sum rate results. When NOMA is not feasible, TDMA with equal time allocation between the two users is adopted. In both figures, the performance of the limited feedback system approaches that of the full CSI system as the feedback rate increases. 

\begin{figure}[ht]
    \centering
    \includegraphics[width=0.75\linewidth]{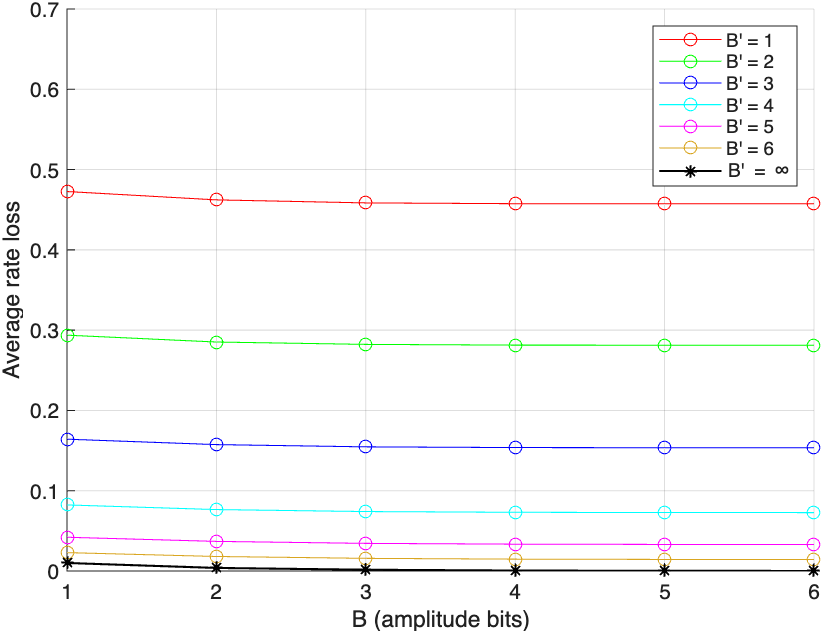}
    \caption{The rate loss for different $B$ values when $B'$ is fixed.}
    \label{fig:4}
\end{figure}

\begin{figure}[ht]
    \centering
    \includegraphics[width=0.75\linewidth]{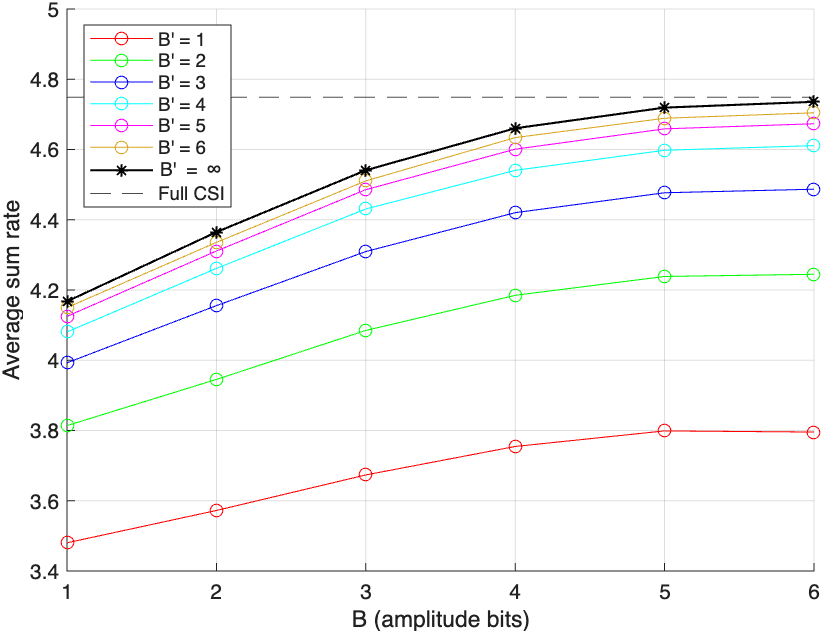}
    \caption{The average sum rate for different $B$  values when $B'$ is fixed.}
    \label{fig:5}
\end{figure}

To further explore the performance of the proposed limited feedback MISO-NOMA framework under different system settings, and to further illustrate the impact of feedback rates, we provide additional simulation results beyond the baseline settings. Specifically, we provide simulation results with three different settings, \(\mathrm{SNR}=15\) dB, \(R_{\rm th}=0.5\) bits/s/Hz and \(N_t=3\) transmit antennas, in Figs.~\ref{fig:15}, \ref{fig:Rth}, and \ref{fig:Nt3}, respectively.

\begin{figure*}[ht]
    \centering
    {\captionsetup[subfloat]{font=footnotesize}
    \subfloat[]{
        \includegraphics[width=0.4\textwidth]{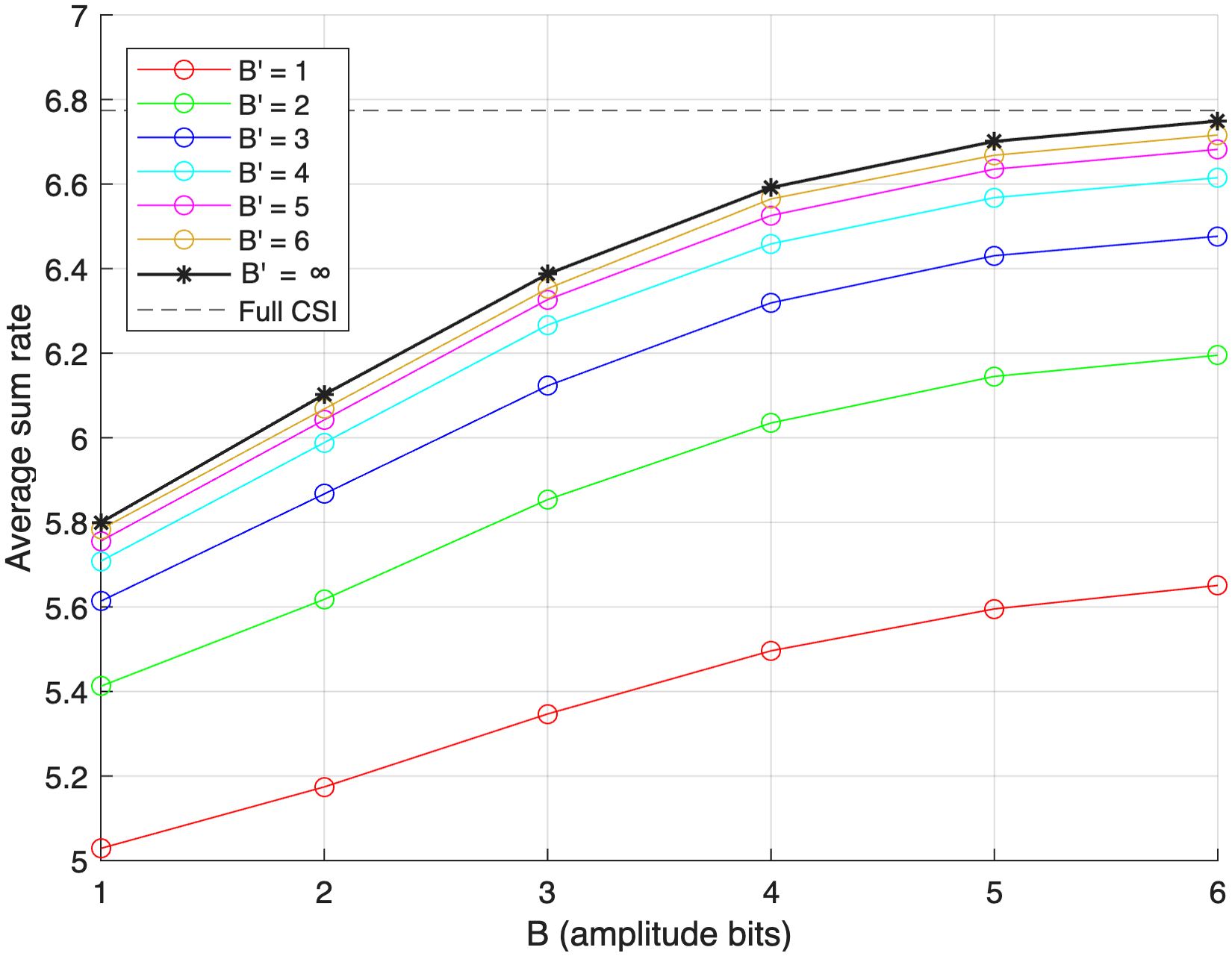}
        \label{fig:aver15}
    }
    \hspace{0.1\textwidth}
    \subfloat[]{
        \includegraphics[width=0.4\textwidth]{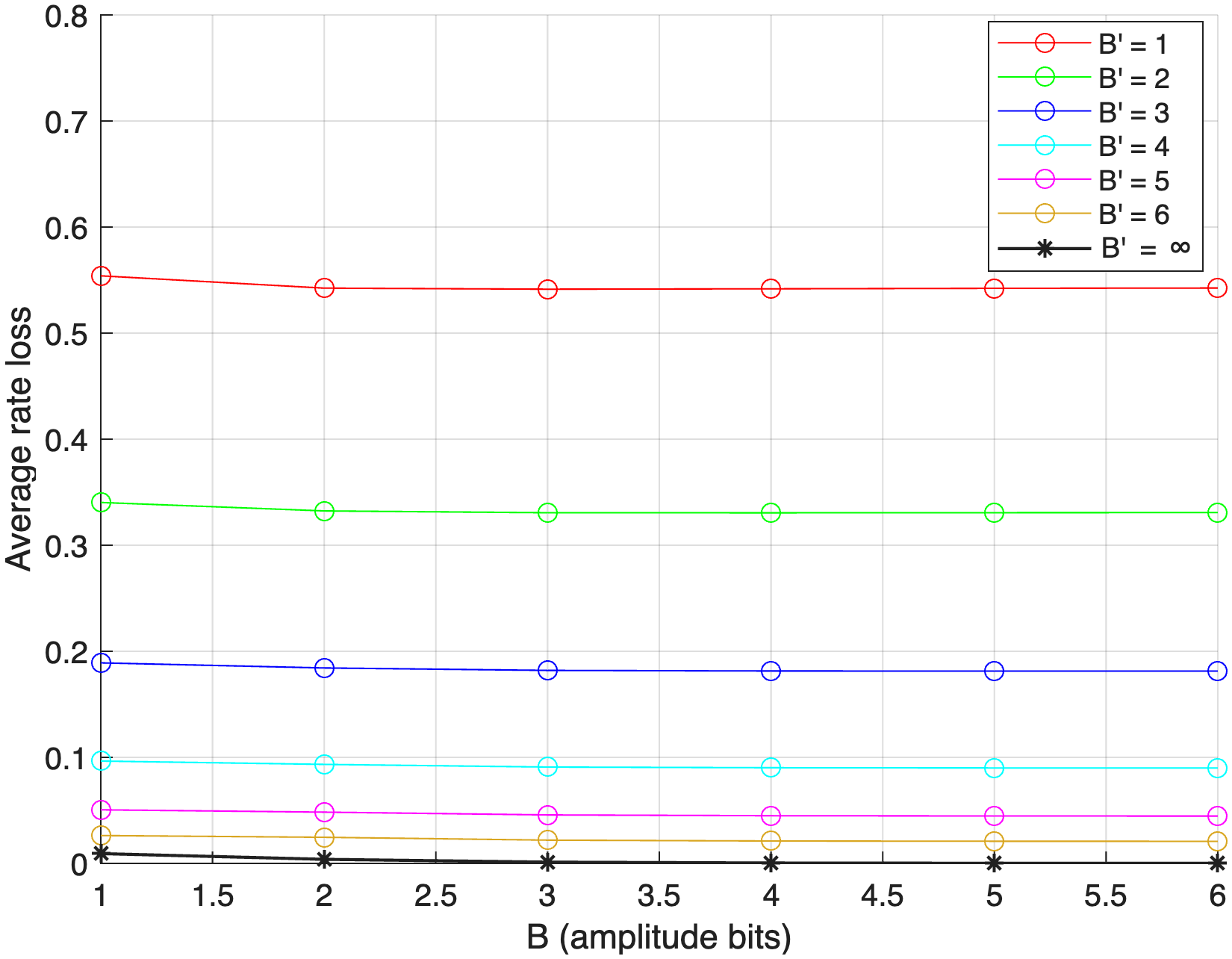}
        \label{fig:rateloss15}
    }}
    \caption{Simulation results under \(\mathrm{SNR}=15\) dB with \(N_t=2\) and \(R_{\rm th}=1\) bit/s/Hz, while other parameters are kept the same as in the baseline setting: (a) average sum rate and (b) rate loss.}
    \label{fig:15}
\end{figure*}

\begin{figure*}[ht]
    \centering
    {\captionsetup[subfloat]{font=footnotesize}
    \subfloat[]{
        \includegraphics[width=0.4\textwidth]{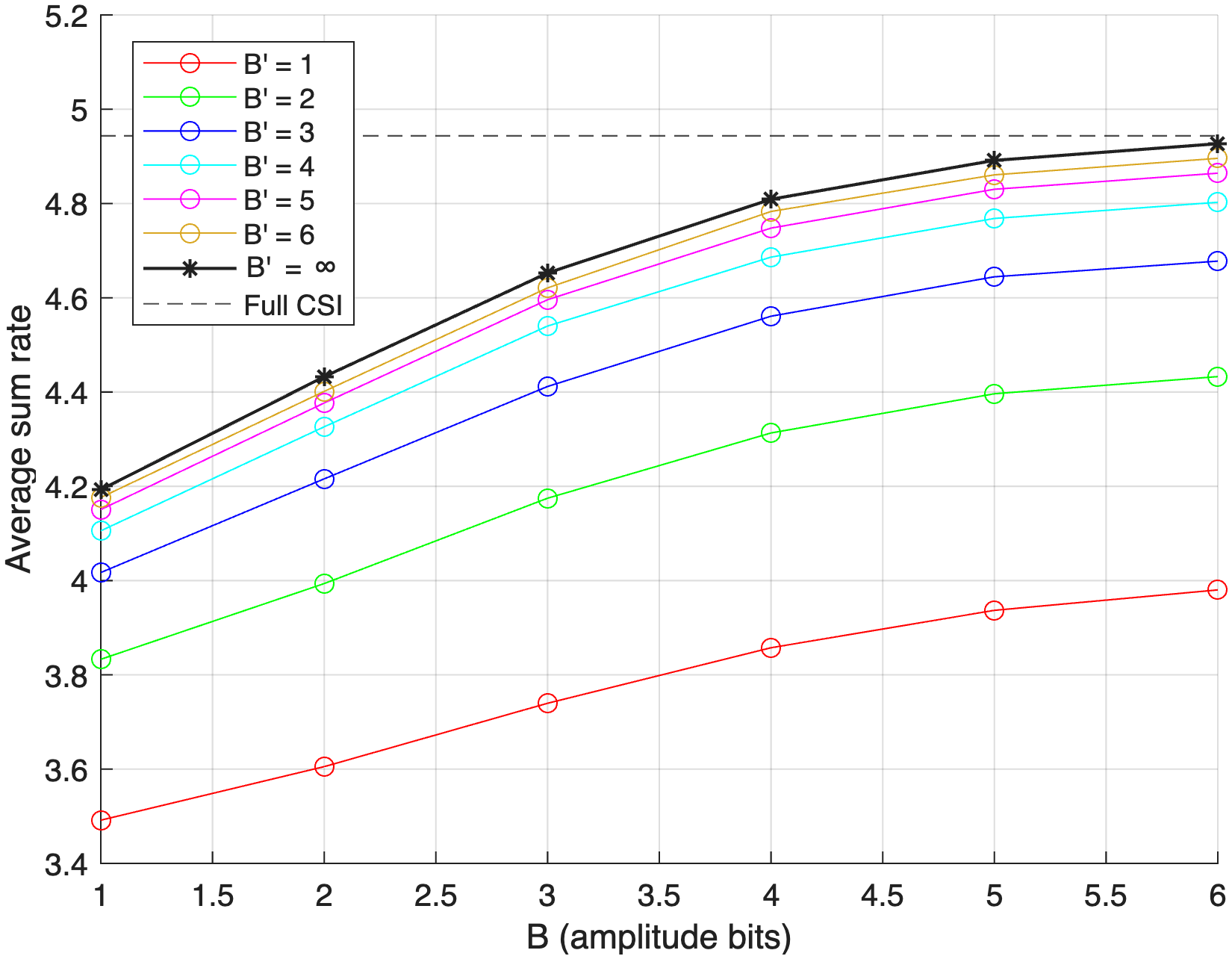}
        \label{fig:averRth}
    }
    \hspace{0.1\textwidth}
    \subfloat[]{
        \includegraphics[width=0.4\textwidth]{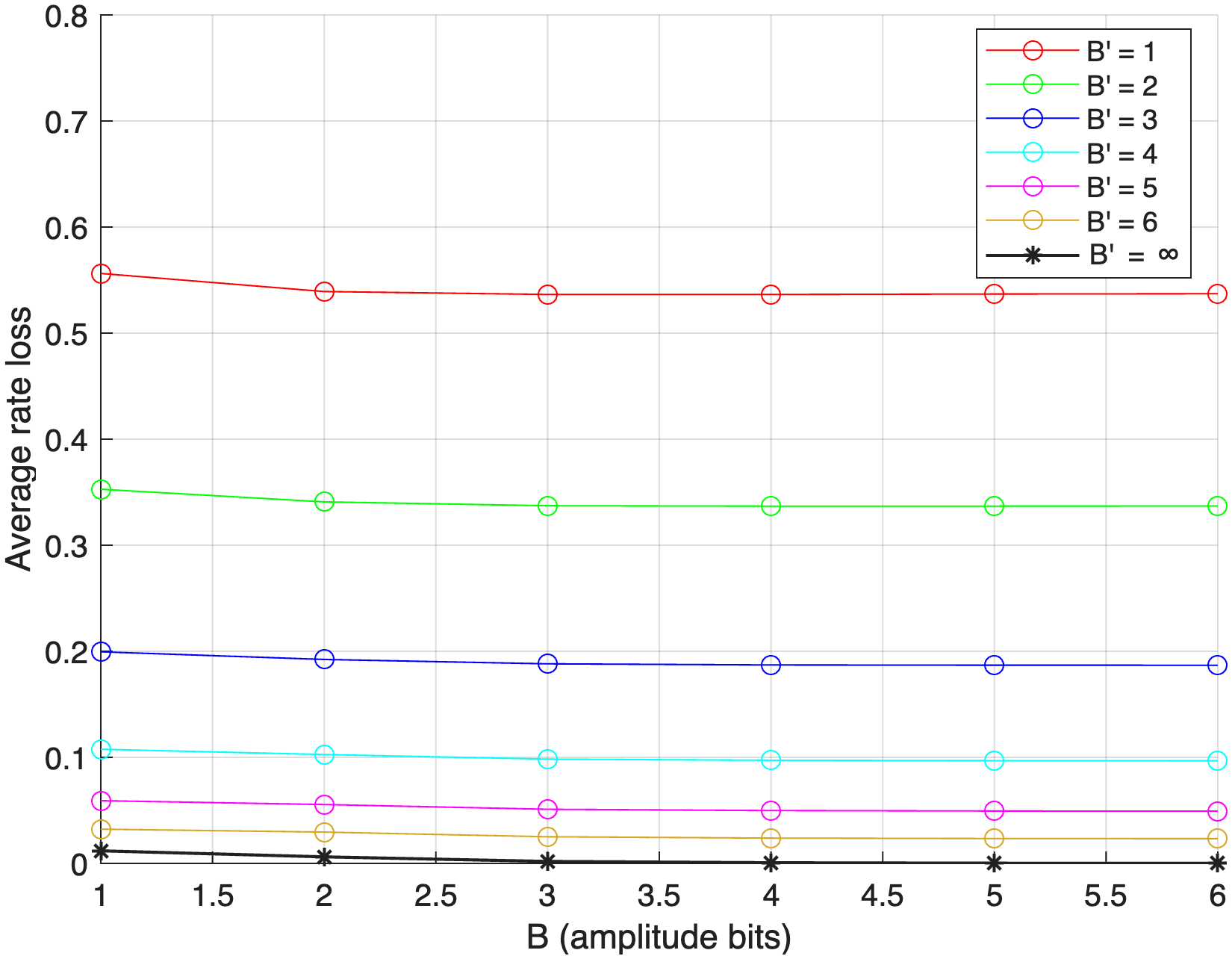}
        \label{fig:ratelossRth}
    }}
    \caption{Simulation results under \(R_{\rm th}=0.5\) bit/s/Hz with \(N_t=2\) and \(\mathrm{SNR}=10\) dB, while other parameters are kept the same as in the baseline setting: (a) average sum rate and (b) rate loss.}
    \label{fig:Rth}
\end{figure*}

\begin{figure*}[ht]
    \centering
    {\captionsetup[subfloat]{font=footnotesize}
    \subfloat[]{
        \includegraphics[width=0.4\textwidth]{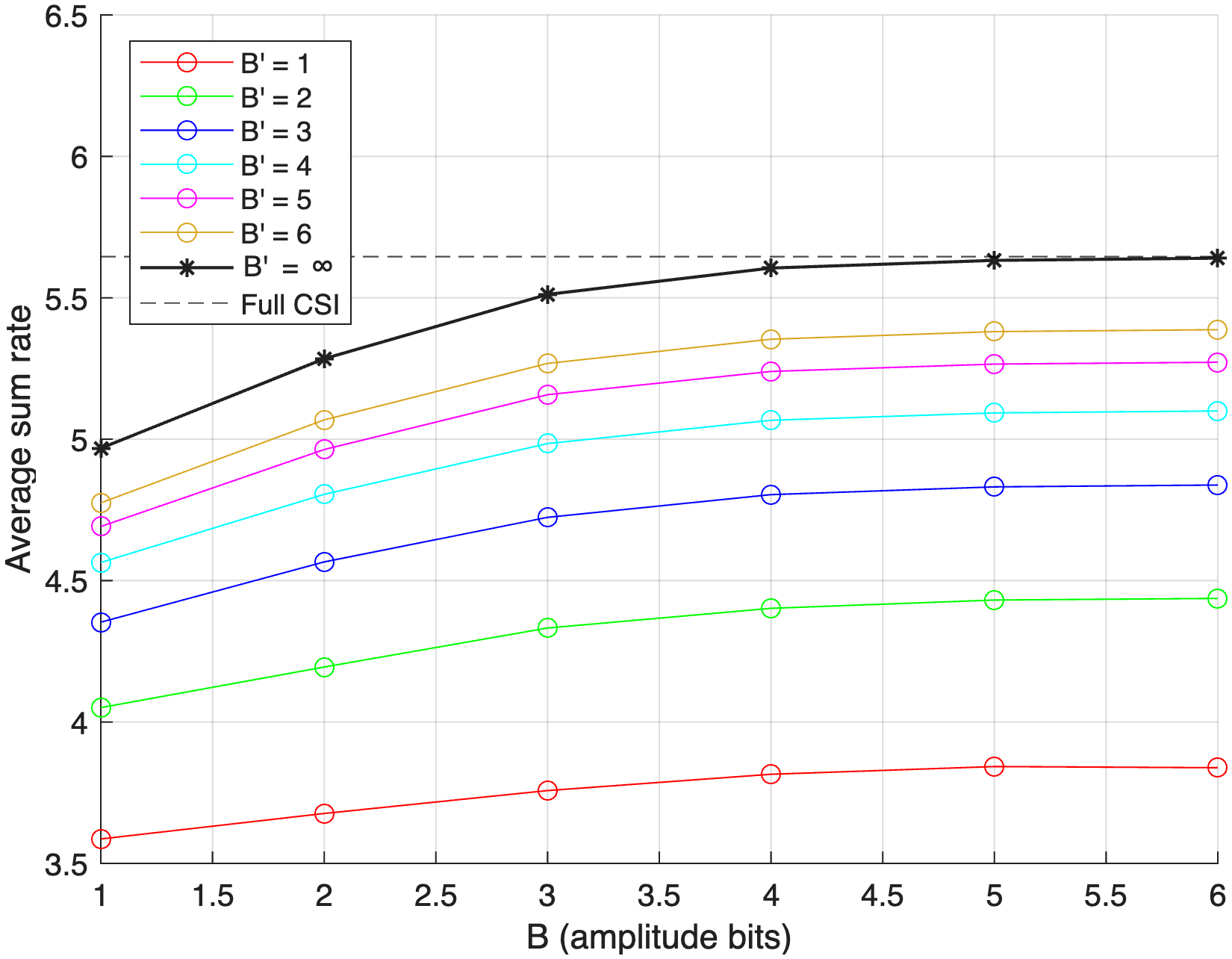}
        \label{fig:averNt3}
    }
    \hspace{0.1\textwidth}
    \subfloat[]{
        \includegraphics[width=0.4\textwidth]{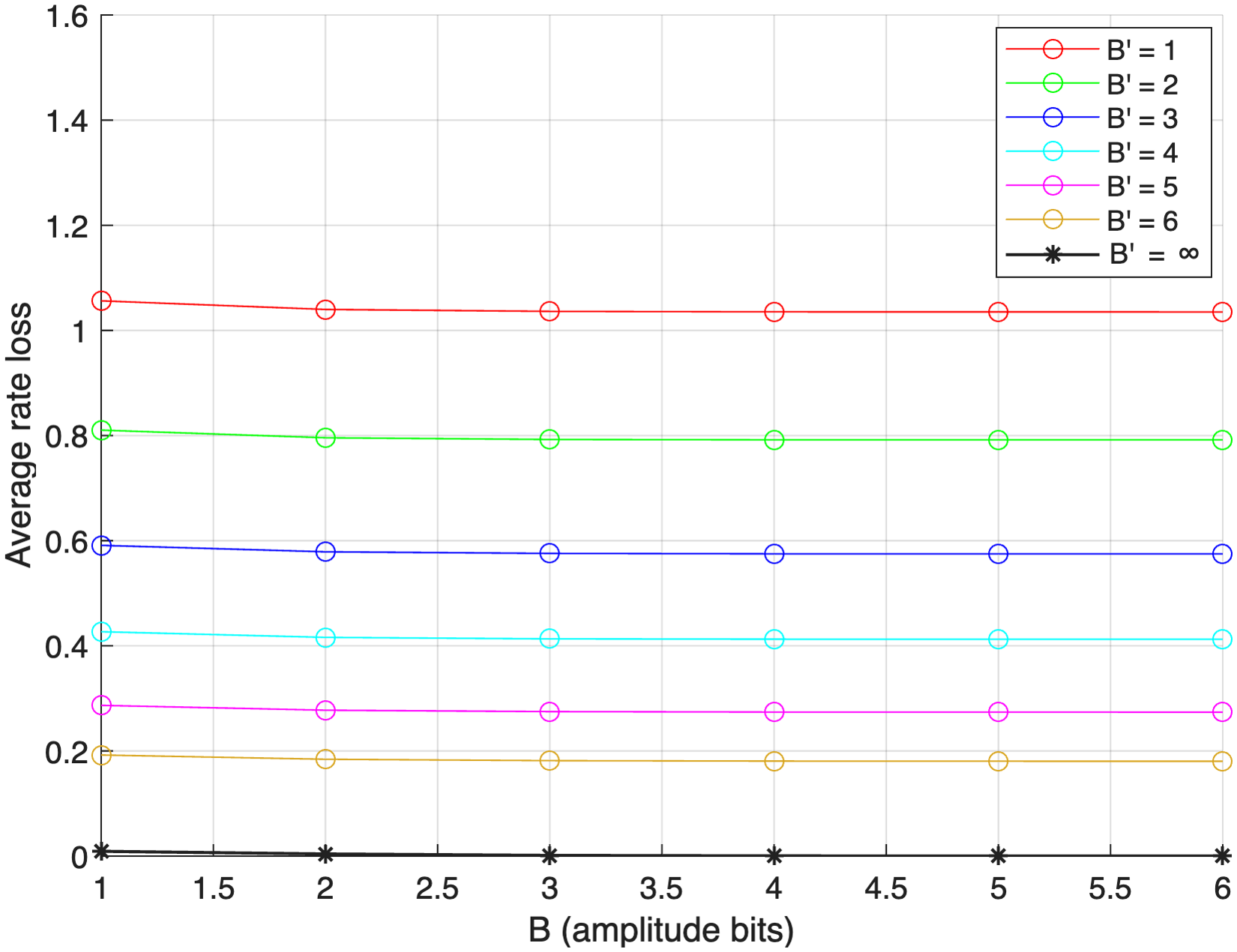}
        \label{fig:ratelossNt3}
    }}
    \caption{Simulation results with \(N_t=3\) transmit antennas, \(R_{\rm th}=1\) bit/s/Hz and \(\mathrm{SNR}=10\) dB, the magnitude quantizers are re-optimized to account for the change in the effective gain distribution: (a) average sum rate and (b) rate loss.}
    \label{fig:Nt3}
\end{figure*}

From Figs.~\ref{fig:15} and \ref{fig:Rth}, we observe that the average sum rate and rate loss show trends similar to the baseline setting as $B$ and $B'$ vary. For the case with $N_t=3$, the magnitude quantizers are re-optimized to account for the change in the effective channel gain distribution. Since the number of transmit antennas increases, a higher $B'$ is generally required to accurately quantize the channel directions. As shown in Fig.~\ref{fig:Nt3}, the performance gap between the case with $B'=6$ and the case with $B'=\infty$ is more noticeable compared to the baseline setting in Fig.~\ref{fig:5}.

\section{Conclusion}\label{conclusion}
This work presented a limited feedback framework design for MISO-NOMA, which is achieved by feeding back the precoding vector's index and the effective gain. We considered the effect of beamforming direction on the SIC order and power allocation. We analyzed NOMA's feasible region and rate loss in limited feedback, providing a rate loss upper bound, which converges to 0 at least exponentially. Simulation results show that the rate loss decreases as the number of feedback bits $B$ and $B'$ increases. Extending the analysis to more users is a possible future research direction.

\appendices
\section{Proof of Theorem~\ref{thm:1}}\label{app:thm1}
We divide the rate loss into two parts.
\begin{equation}
    \begin{aligned}
    \Delta R&=(R_{\mathrm N}(H;\mathbf{h}/\|\mathbf{h}\|;\beta^*)
    -R_{\mathrm N}(H;\mathbf{h}/\|\mathbf{h}\|;\beta_q))\\
    &+(R_{\mathrm N}(H;\mathbf{h}/\|\mathbf{h}\|;\beta_q)
    -R_{\mathrm N}(H;\mathbf{w};\beta_q)).
    \end{aligned}
\label{eq:A1}
\end{equation}

The first item describes the rate loss caused by inaccurate power allocation and is decided by both $\eta_{ii}$ and $\delta$.

\begin{lemma}\label{le:3}
$\Delta R_1 \triangleq R_{\mathrm N}(H;\mathbf{h}/\|\mathbf{h}\|;\beta^*)
    -R_{\mathrm N}(H;\mathbf{h}/\|\mathbf{h}\|;\beta_q)$ is bounded by:
\begin{equation}
    \begin{aligned}
    \Delta R_1
    & \le \frac{P}{\ln2}\left(\frac{ H_{1}}{\sigma^2+\beta^* PH_{1}} +\frac{H^*_{21}}{\sigma^2} + \frac{H_{2} - H_{21}^*}{\sigma^2 +P H_{21}^*}\right) |\Delta \beta |. 
\end{aligned}
\label{eq:A2}
\end{equation}
\end{lemma}
\noindent\textit{Proof}. The proof is provided in Appendix ~\ref{app:lemma3}. 

The second part captures the rate loss caused by limited beamforming direction choices.
\begin{lemma}\label{le:4}
$\Delta R_2 \triangleq R_{\mathrm N}(H;\mathbf{h}/\|\mathbf{h}\|;\beta_q)
    -R_{\mathrm N}(H;\mathbf{w};\beta_q)$ is bounded by:
\begin{equation}
    \begin{aligned}
        \Delta R_2
        \le &\frac{1}{\ln 2} \left[\frac{2(1-\beta_q)P^2 \beta_q \eta_{22}H_2^2 \sqrt{2 - 2\eta_{11}}}{\sigma^4} \right.\\
        & \left.+\frac{(1-\beta_q)PH_2(1-\eta_{22})}{\sigma^2} + \frac{\beta_q P H_1(1-\eta_{11})}{\sigma^2}\right].
    \end{aligned}
\label{eq:A3}
\end{equation}
\end{lemma}
\noindent\textit{Proof}. The proof is provided in Appendix~\ref{app:lemma4}.

Combining Eq.~\eqref{eq:A2} and Eq.~\eqref{eq:A3}, and since $0<\beta_q<1$, $0<\beta_q(1-\beta_q)<\frac{1}{4}$, $H_1, H_{21}^*>0$, we have
\begin{equation}
\begin{aligned}
   \Delta R
    & \le \frac{P}{\ln 2} \left[\frac{H_1(1-\eta_{11})}{\sigma^2} +  \frac{P \eta_{22}H_2^2 \sqrt{2 - 2\eta_{11}}}{2\sigma^4} \right.\\  
        & \quad\left. + \frac{H_2(1-\eta_{22})}{\sigma^2}+(\frac{ H_{1}}{\sigma^2} +\frac{H^*_{21}}{\sigma^2} + \frac{H_{2} - H_{21}^*}{\sigma^2 }) |\Delta \beta | \right]\\
    & \triangleq C_1H_1(1-\eta_{11})+C_2H_2^2\sqrt{2-2\eta_{11}}+C_3H_2(1-\eta_{22}) \\
    &  \qquad+C_4(H_1+H_2)|\Delta \beta|.
\end{aligned}
\end{equation}
This completes the proof of Theorem~\ref{thm:1}.

\section{Proof of Lemma~\ref{le:1}}\label{app:lemma1}
In the MISO system, we have the sum rate:
\begin{equation}
   \begin{aligned}
    R_{sum} &=\log\bigl(1+\frac{\beta |\mathbf{h}_1^H\mathbf{w}_1|^2}{\sigma^2}\bigr) + \log\bigl(1+ \frac{(1-\beta) |\mathbf{h}_2^H\mathbf{w}_2|^2}{\beta|\mathbf{h}_2^H\mathbf{w}_1|^2+\sigma^2}\bigr).
\end{aligned} 
\end{equation}
To find the optimal power allocation, we calculate the partial derivative, 
\begin{equation}
    \begin{aligned}
        \frac{\partial R_{sum}}{\partial \beta} =& \frac{1}{\ln 2}\left[\frac{\sigma^2(|\mathbf{h}_1^H\mathbf{w}_1|^2-|\mathbf{h}_2^H\mathbf{w}_1|^2)}{(\beta |\mathbf{h}_1^H\mathbf{w}_1|^2+\sigma^2)((\beta |\mathbf{h}_2^H\mathbf{w}_1|^2+\sigma^2)}\right.\\
        &-\left.\frac{|\mathbf{h}_2^H\mathbf{w}_2|^2-|\mathbf{h}_2^H\mathbf{w}_1|^2}{\beta |\mathbf{h}_2^H\mathbf{w}_1|^2+\sigma^2+ (1-\beta)|\mathbf{h}_2^H\mathbf{w}_2|^2}\right].
    \end{aligned}
\end{equation}
Since $\mathbf{w}_2$ is selected using MRT, $\mathbf{w}_2=\arg\max_{\mathbf{w}\in C}|\mathbf{h}_2^H\mathbf{w}|$. We have $|\mathbf{h}_1^H\mathbf{w}_1|> |\mathbf{h}_2^H\mathbf{w}_2|> |\mathbf{h}_2^H\mathbf{w}_1|$. 
The stationary point holds when
\begin{align}
    \frac{\sigma^2(H_{11}-H_{21})}{(\beta H_{11}+\sigma^2)(\beta H_{21}+\sigma^2)}= \frac{H_{22}-H_{21}}{\beta H_{21}+\sigma^2+ (1-\beta)H_{22}}.
\end{align}
This is a quadratic equation in a single variable and has at most one root $\beta_0$ in the interval (0,1). So, the optimal power allocation must be one of the two boundary points or the stationary point $\beta_0$ if it lies in the feasible interval.

\section{Proof of Lemma~\ref{le:2}}\label{app:lemma2}
\begin{lemma}
    $\Delta \beta_{max}$, $\Delta \beta_{min}$ and $\Delta \beta_0$ all converge to 0 at least exponentially when $B$ and $B'$ increase.
\label{lemma:5}
\end{lemma}

The proof of Lemma~\ref{lemma:5} is provided in Appendix~\ref{app:lemma5}. 
We now proceed to prove Lemma~\ref{le:2}. Switching from full CSI to limited feedback, there is a chance that the choice among $\{\beta_{\min},\beta_0,\beta_{\max}\}$ changes. For example, when $\beta^*= \beta^*_{\min}$, it is possible that the best power allocation $\beta_q$ to maximize the sum rate is $\beta_{0,q}$ not $\beta_{\min,q}$. However, since $R^{LF}(H;\mathbf{w};\beta_{0,q})>R^{LF}(H;\mathbf{w};\beta_{\min,q})$, we have $R^{Full}(H;\mathbf{h}/\|\mathbf{h}\|;\beta_{\min})-R^{LF}(H;\mathbf{w};\beta_{\min,q})>R^{Full}(H;\mathbf{h}/\|\mathbf{h}\|;\beta_{\min})-R^{LF}(H;\mathbf{w};\beta_{0,q})$. Calculating the rate loss, assuming that the choice among $\{\beta_{\min},\beta_0,\beta_{\max}\}$ in both full CSI and limited feedback is the same, provides an upper bound for the real rate loss. The above reasoning together with Lemma~\ref{lemma:5} completes the proof.

\section{Proof of Lemma~\ref{le:3}}\label{app:lemma3}

We can write $\Delta R_1$ as the sum of the rate losses for the strong and weak users:
\begin{equation}
\begin{aligned}
    &\Delta R_1 = R_{\mathrm N}(H;\mathbf{h}/\|\mathbf{h}\|;\beta^*) -R_{\mathrm N}(H;\mathbf{h}/\|\mathbf{h}\|;\beta_q)\\
    & = \log\left(\frac{\sigma^2+\beta^*PH_{1}}{\sigma^2+\beta_qPH_{1}}\right)+\log\left(\frac{1+\frac{(1-\beta^*)PH_{2}}{\beta^*PH^*_{21}+\sigma^2}}{1+\frac{(1-\beta_q)PH_{2}}{\beta_q PH_{21}^*+\sigma^2}}\right)\\
    & \triangleq \Delta R_{1,s} +\Delta R_{1,w}
\end{aligned}
\end{equation}

\noindent\textbf{Step 1:} Rate loss of the strong user ($\Delta R_{1,s}$):
we express this term as a function of $\Delta \beta$
\begin{equation}
    \begin{aligned}
        |\Delta R_{1,s}| &= \left|\log\left(\frac{\sigma^2+\beta_q PH_{1}}{\sigma^2+\beta^* P H_{1}}\right)\right|\\
        & = \left|\log\left(1+\frac{\Delta \beta P H_{1}}{\sigma^2+\beta^* PH_{1}}\right)\right|\\
        & \le \frac{1}{\ln 2}\frac{ P H_{1}}{\sigma^2+\beta^* PH_{1}} \left|\Delta \beta\right|
    \end{aligned}
\end{equation}

\noindent\textbf{Step 2:} Rate loss of the weak user ($\Delta R_{1,w}$): we write the rate of the weak user as a function of power allocation $\beta$:
\begin{equation}
    R_{1,w}(\beta) = \log(\beta P(H_{21}^*-H_{2})+\sigma^2+P H_{2})-\log(\beta P H_{21}^*+\sigma^2)
\end{equation}
From intermediate value theorem, we know that there exists a $\xi$ between  $\beta^*$ and $\beta_q$, such that
\begin{equation}
    |\Delta R_{1,w}(\beta)| = |R_{1,w}(\beta^*)-R_{1,w}(\beta_q)|= | R'_{1,w}(\xi)\Delta \beta|
\label{eq:17}
\end{equation}
To get the full expression, we calculate the first-order derivative
\begin{equation}
\begin{aligned}
    R'_{1,w}(\beta)&= \frac{1}{\ln2}\left(\frac{P(H_{21}^*-H_{2})}{\beta P (H_{21}^*-H_{2})+\sigma^2+PH_{2}}\right.\\
    &\quad\left.-\frac{PH_{21}^*}{\beta P H_{21}^*+\sigma^2}\right).
\end{aligned}
\label{eq:18}
\end{equation}
Also, since $H_2\ge H^*_{21}\ge 0$ and $0<\xi<1$, we have $\xi P H_{21}^*+\sigma^2\ge \sigma^2>0$, $\xi P (H_{21}^*-H_{2})+\sigma^2+PH_{2}\ge P (H_{21}^*-H_{2})+\sigma^2+PH_{2}=P H_{21}^*+\sigma^2>0$. Using these inequalities and substituting Eq.~\eqref{eq:18} into Eq.~\eqref{eq:17}, we have an upper bound of $\Delta R_{1,w}$ as a function of $\Delta\beta$  as follows.

\begin{equation}
\begin{aligned}
    |\Delta R_{1,w}(\beta)|
    &= | R'_{1,w}(\xi)\Delta \beta|\\
    & = \frac{P}{\ln2}\left(\frac{H^*_{21}}{\xi P H_{21}^*+\sigma^2} \right.\\
    &\quad\left. + \frac{H_{2} - H_{21}^*}{\xi P (H_{21}^*-H_{2})+\sigma^2+PH_{2}}\right) |\Delta \beta |\\
    & \le \frac{P}{\ln2}\left(\frac{H^*_{21}}{\sigma^2} + \frac{H_{2} - H_{21}^*}{\sigma^2 +P H_{21}^*}\right) |\Delta \beta |. 
\end{aligned}
\label{eq:19}
\end{equation}
As a result, we have the rate loss bound as a function of $\Delta \beta$:
\begin{equation}
    \begin{aligned}
        \Delta R_1 &= \Delta R_{1,s} + \Delta R_{1,w}\\
        & \le |\Delta R_{1,s}| + |\Delta R_{1,w}|\\
        & \le \frac{P}{\ln2}\left(\frac{ H_{1}}{\sigma^2+\beta^* PH_{1}} +\frac{H^*_{21}}{\sigma^2} + \frac{H_{2} - H_{21}^*}{\sigma^2 +P H_{21}^*}\right) |\Delta \beta | 
    \end{aligned}
    \label{eq:2nd}
\end{equation}
This completes the proof of Lemma~\ref{le:3}.

\section{Proof of Lemma~\ref{le:4}}\label{app:lemma4}

The upper bound of $\Delta R_2$ can be decomposed into the rate loss of the strong user and that of the weak user.
\begin{equation}
    \begin{aligned}
    &R_{\mathrm N}(H;\mathbf{h}/\|\mathbf{h}\|;\beta)
    -R_{\mathrm N}(H;\mathbf{w};\beta) \\
    &=\left[\log\left(1+\frac{\beta P H_1}{\sigma^2}\right)-\log\left(1+\frac{\beta P \eta_{11}H_1}{\sigma^2}\right)\right]\\
    &\quad + \left[\log\left(1+\frac{(1-\beta)PH_2}{\beta P\left|\mathbf{h}_2^H \frac{\mathbf{h}_1}{\|\mathbf{h}_1\|}\right|^2+\sigma^2}\right)\right.\\
    &\quad\left. - \log\left(1+\frac{(1-\beta)P\eta_{22}H_2}{\beta P H_{21}+\sigma^2}\right)\right]\\
    &\triangleq \Delta R_{2,s}+\Delta R_{2,w}
\end{aligned}
\end{equation}

\noindent\textbf{Step 1:} Rate loss of the strong user
\begin{equation}
    \begin{aligned}
        \Delta R_{2,s} &= \log\left(1+\frac{\beta P H_1}{\sigma^2}\right)-\log\left(1+\frac{\beta P \eta_{11}H_1}{\sigma^2}\right)\\
        &=\log\left(\frac{\sigma^2+\beta P H_1}{\sigma^2 + \beta P\eta_{11}H_1}\right)\\
        &=\log\left(1+\frac{\beta P H_1(1-\eta_{11})}{\sigma^2+\beta P\eta_{11}H_1}\right)\\
        &\le\frac{1}{\ln 2} \frac{\beta P H_1(1-\eta_{11})}{\sigma^2+\beta P\eta_{11}H_1}
    \end{aligned}
    \label{eq:s}
\end{equation}

\noindent\textbf{Step 2:} Rate loss of the weak user
\begin{equation}
\begin{aligned}
        \Delta R_{2,w} &= \log\left(1+\frac{(1-\beta)PH_2}{\beta P\left|\mathbf{h}_2^H \frac{\mathbf{h}_1}{\|\mathbf{h}_1\|}\right|^2+\sigma^2}\right)\\
    &\quad - \log\left(1+\frac{(1-\beta)P\eta_{22}H_2}{\beta P H_{21}+\sigma^2}\right)
\end{aligned}
\end{equation}
Let $H_{21}^*=\left|\mathbf{h}_2^H \frac{\mathbf{h}_1}{\|\mathbf{h}_1\|}\right|^2$, applying the Lagrange mean value theorem on $\log(1+x)$, we have
\begin{equation}
\begin{aligned}
        \Delta R_{2,w} & \le \frac{1}{\ln2 \left(1+\frac{(1-\beta)P\eta_{22}H_2}{\beta PH_{21}^*+\sigma^2}\right)} \\
        &\quad\cdot\left|\frac{(1-\beta)PH_2}{\beta PH_{21}^*+\sigma^2}- \frac{(1-\beta)P\eta_{22}H_2}{\beta P H_{21}+\sigma^2}\right|
\end{aligned}
\end{equation}
We can write the right hand side as a function of $\left|H_{21}^*-H_{21}\right|$ and $\eta_{22}$,
\begin{equation}
\begin{aligned}
    &\left|\frac{(1-\beta)PH_2}{\beta PH_{21}^*+\sigma^2}- \frac{(1-\beta)P\eta_{22}H_2}{\beta P H_{21}+\sigma^2}\right|\\
    & =\left|\frac{(1-\beta)PH_2(1-\eta_{22})}{\beta PH_{21}^*+\sigma^2}+ \frac{(1-\beta)PH_2\eta_{22}}{\beta PH_{21}^*+\sigma^2} \right.\\
    &\quad\left. - \frac{(1-\beta)P\eta_{22}H_2}{\beta P H_{21}+\sigma^2}\right|\\
    &= (1-\beta)P\eta_{22}H_2 \frac{\beta P \left|H_{21}-H_{21}^*\right|}{(\beta PH_{21}^*+\sigma^2)(\beta P H_{21}+\sigma^2)} \\
    &\quad + \frac{(1-\beta)PH_2(1-\eta_{22})}{\beta PH_{21}^*+\sigma^2}
\end{aligned}  
\label{eq:w}
\end{equation}

\begin{lemma}
$\left|H_{21}^*-H_{21}\right| $ is determined by $H_2$ and $\eta_{11}$:
    \begin{equation}
    \begin{aligned}
        \left|H_{21}^*-H_{21}\right| 
        & \le 2 H_2 \sqrt{2 - 2\eta_{11}}
    \end{aligned}
\end{equation}
\label{lemma:S1}
\end{lemma}
\noindent\textit{Proof.} See Appendix~\ref{app:S1}.

Applying Lemma~\ref{lemma:S1} to Eq.~\eqref{eq:w} and adding Eq.~\eqref{eq:s}, we have
\begin{equation}
    \begin{aligned}
    & R_{\mathrm N}(H;\mathbf{h}/\|\mathbf{h}\|;\beta)
    -R_{\mathrm N}(H;\mathbf{w};\beta) \\
    & \le 
    \frac{1}{\ln2 \left(1+\frac{(1-\beta)P\eta_{22}H_2}{\beta PH_{21}^*+\sigma^2}\right)} \left|\frac{(1-\beta)PH_2(1-\eta_{22})}{\beta PH_{21}^*+\sigma^2} \right.\\
    &\quad \left.+ (1-\beta)P\eta_{22}H_2 \frac{2\beta P H_2 \sqrt{2 - 2\eta_{11}}}{(\beta PH_{21}^*+\sigma^2)(\beta P H_{21}+\sigma^2)}\right|\\
    & \quad + \frac{1}{\ln 2} \frac{\beta P H_1(1-\eta_{11})}{\sigma^2+\beta P\eta_{11}H_1}
\end{aligned}
\label{eq:16}
\end{equation}
Eq.~\eqref{eq:16} provides an upper bound for the rate loss term based on the beamforming direction, and it approaches 0 as $\eta_{11}$ and $\eta_{22}$ approach 1.

Let $\beta=\beta_q$, we have
\begin{equation}
    \begin{aligned}
    & R_{\mathrm N}(H;\mathbf{h}/\|\mathbf{h}\|;\beta_q)
    -R_{\mathrm N}(H;\mathbf{w};\beta_q) \\
    & \le 
    \frac{1}{\ln2 \left(1+\frac{(1-\beta_q)P\eta_{22}H_2}{\beta_q P H_{21}^*+\sigma^2}\right)} \left|\frac{(1-\beta_q)PH_2(1-\eta_{22})}{\beta_q PH_{21}^*+\sigma^2} \right.\\
    &\quad \left.+ (1-\beta_q)P\eta_{22}H_2 \frac{2\beta_q P H_2 \sqrt{2 - 2\eta_{11}}}{(\beta_q P H_{21}^*+\sigma^2)(\beta_q P H_{21}+\sigma^2)}\right|\\
    &\quad +\frac{1}{\ln 2} \frac{\beta_q P H_1(1-\eta_{11})}{\sigma^2+\beta_q P\eta_{11}H_1}
\end{aligned}
\end{equation}
By simple scaling, we have the equation in Lemma~\ref{le:4}:
\begin{equation}
    \begin{aligned}
    & R_{\mathrm N}(H;\mathbf{h}/\|\mathbf{h}\|;\beta_q)
    -R_{\mathrm N}(H;\mathbf{w};\beta_q) \\
    &\le \frac{1}{\ln 2} \left(\frac{(1-\beta_q)PH_2(1-\eta_{22})}{\sigma^2} \right.\\
    &\quad\left. +\frac{\beta_q P H_1(1-\eta_{11})}{\sigma^2} +  \frac{2(1-\beta_q)P^2 \beta_q \eta_{22}H_2^2 \sqrt{2 - 2\eta_{11}}}{\sigma^4}\right).
\end{aligned}
\end{equation}
The first term approaches $0$ when $\eta_{22}\to 1$, and the second and third terms decrease exponentially to $0$ when $\eta_{11}\to 1$.

\section{Proof of Lemma~\ref{lemma:S1}}\label{app:S1}
Since $\|\mathbf{w}_1\|=1$, we know $\left|\mathbf{w}_1+\frac{\mathbf{h}_1}{\|\mathbf{h}_1\|}\right|\le 2$. By factoring out $\|\mathbf{h}_2\|$ and applying Cauchy-Schwarz inequality, we have
\begin{equation}
    \begin{aligned}
        &\left|H_{21}^*-H_{21} \right|\\
        & = \left| \left|\mathbf{h}_2^H\frac{\mathbf{h}_1}{\|\mathbf{h}_1\|}\right|^2-\left|\mathbf{h}_2^H\mathbf{w}_1\right|^2 \right|\\
        & = \left|(\mathbf{h}_2^H\frac{\mathbf{h}_1}{\|\mathbf{h}_1\|}+\mathbf{h}_2^H\mathbf{w}_1)(\mathbf{h}_2^H\frac{\mathbf{h}_1}{\|\mathbf{h}_1\|}-\mathbf{h}_2^H\mathbf{w}_1)\right|\\
        &\le \|\mathbf{h}_2\|^2 \left|\mathbf{w}_1+\frac{\mathbf{h}_1}{\|\mathbf{h}_1\|}\right|\left|\mathbf{w}_1-\frac{\mathbf{h}_1}{\|\mathbf{h}_1\|}\right|\\
        & \le 2 \|\mathbf{h}_2\|^2 \left|\mathbf{w}_1-\frac{\mathbf{h}_1}{\|\mathbf{h}_1\|}\right|
    \end{aligned}
    \label{eq:s1}
\end{equation}
By expanding the squared Euclidean norm, we obtain
\begin{equation}
    \begin{aligned}
        \left|\mathbf{w}_1-\frac{\mathbf{h}_1}{\|\mathbf{h}_1\|}\right|^2 
        & = \left(\mathbf{w}_1-\frac{\mathbf{h}_1}{\|\mathbf{h}_1\|}\right)^H \left(\mathbf{w}_1-\frac{\mathbf{h}_1}{\|\mathbf{h}_1\|}\right)\\
        & = 2 - 2\mathcal{R}\{\mathbf{w}_1^H \frac{\mathbf{h}_1}{\|\mathbf{h}_1\|}\}\\
        & = 2 - 2\sqrt{\eta_{11}}
    \end{aligned}
\end{equation}
Bringing this back to Eq.~\eqref{eq:s1}, and using the property $0<\eta_{11}\le1$, we have
\begin{equation}
    \begin{aligned}
        \left|H_{21}^*-H_{21}\right| 
        & \le 2 \|\mathbf{h}_2\|^2 \left|\mathbf{w}_1-\frac{\mathbf{h}_1}{\|\mathbf{h}_1\|}\right|\\
        & = 2 \|\mathbf{h}_2\|^2 \sqrt{2 - 2\sqrt{\eta_{11}}}\\
        & \le 2 H_2 \sqrt{2 - 2\eta_{11}}
    \end{aligned}
\end{equation}
This completes the proof of Lemma~\ref{lemma:S1}.

\section{Proof of Lemma~\ref{lemma:5}}\label{app:lemma5}

The optimal power allocation with full CSI has three possibilities $\beta^*= \{\beta_{\min},\beta_0,\beta_{\max}\}$. Here, we present the upper bound of $\Delta\beta$ in each case, express them as functions of $\eta$ and $\delta$, and prove the convergence rate.

\noindent\textbf{Step 1:} when $\beta^*=\beta_{\min}$, we have
\begin{equation}
    \begin{aligned}
        |\Delta \beta| &= |\beta_q -\beta^*|\\
        &=\left|\frac{\epsilon\sigma^2}{P \hat{H}_{11}}-\frac{\epsilon\sigma^2}{PH_1}\right|\\
        &\le \frac{\epsilon\sigma^2}{P} \max\left\{\left|\frac{1}{\eta_{11}H_1}-\frac{1}{H}_1\right|,\left|\frac{1}{\eta_{11}H_1+\delta}-\frac{1}{H}_1\right| \right\}\\
        &= \frac{\epsilon\sigma^2}{P}\max\left\{\frac{1-\eta_{11}}{\eta_{11}H_1}, \left|\frac{(1-\eta_{11})H_1-\delta}{(\eta_{11}H_1+\delta)H_1}\right|\right\}\\
    \end{aligned}
\end{equation}

There are two possible cases in the maximization.

\noindent\textit{Case 1:}  When $|\Delta \beta| = \frac{\epsilon\sigma^2}{P}\frac{1-\eta_{11}}{\eta_{11}H_1}$,
\begin{equation}
    \begin{aligned}
        \mathbb{E}\left[\frac{\epsilon\sigma^2}{P}\frac{1-\eta_{11}}{\eta_{11}H_1}\right]=\frac{\epsilon\sigma^2}{P}\mathbb{E}\left[\frac{1}{H_1}\right]\mathbb{E}\left[\frac{1-\eta_{11}}{\eta_{11}}\right]
    \end{aligned}
\end{equation}
Since $H_1$ follows the Gamma distribution, $\mathbb{E}[\frac{1}{H_1}]=\frac{1}{N_t-1}$. We write $\mathbb{E}[\frac{1-\eta_{11}}{\eta_{11}}]$ as sum of two parts, one given $\eta_{11}<t$ and the second part its complement:

\begin{equation}
\begin{aligned}
    \mathbb{E}\left[\frac{1-\eta_{11}}{\eta_{11}}\right]&=\mathbb{E}\left[\left.\frac{1-\eta_{11}}{\eta_{11}}\right|\eta_{11}<t\right]Pr(\eta_{11}<t)\\
    &\quad +\mathbb{E}\left[\left.\frac{1-\eta_{11}}{\eta_{11}}\right|\eta_{11}>t\right]Pr(\eta_{11}>t)
\end{aligned}
\label{equ1minuseta}
\end{equation}
where t is a constant that will be chosen later.

For the first part of \eqref{equ1minuseta}, we have
\begin{equation}
\begin{aligned}
        \mathbb{E}\left[\left.\frac{1-\eta_{11}}{\eta_{11}}\right|\eta_{11}<t\right]Pr(\eta_{11}<t)
        &=\mathbb{E}\left[\frac{1-\eta_{11}}{\eta_{11}}*1_{\eta_{11}<t}\right]\\
        &\le\mathbb{E}\left[\frac{1}{\eta_{11}}*1_{\eta_{11}<t}\right]
\end{aligned}
\end{equation}
For $0<x<1$ and $n\ge1$, it is easy to prove that $1-x^n \le n(1-x)$. As a result, $F_{\eta}(\eta)=\left(1-(1-\eta)^{N_t-1}\right)^{2^{B'}}\le \left((N_t-1)\eta\right)^{2^{B'}}$. We have 
$ f_{\eta}(\eta)=F'_{\eta}(\eta) \le 2^{B'}(N_t-1)^{2^{B'}}\eta^{2^{B'}-1}$.
As a result,
\begin{equation}
    \begin{aligned}
        \mathbb{E}\left[\frac{1}{\eta_{11}}*1_{\eta_{11}<t}\right] &= \int_0^t\frac{1}{x}f_{\eta}(x)dx\\
        & \le \int_0^t\frac{1}{x} 2^{B'}(N_t-1)^{2^{B'}}x^{2^{B'}-1} dx\\
        & = \frac{2^{B'}(N_t-1)^{2^{B'}}}{2^{B'}-1}t^{2^{B'}-1}
    \end{aligned}
\end{equation}
Recall that $t$ is a positive constant introduced for the proof. By choosing $t=1/[2(N_t-1)]$, we have
\begin{equation}
    \begin{aligned}
        \mathbb{E}\left[\frac{1}{\eta_{11}}*1_{\eta_{11}<t}\right]
        & \le \frac{2^{B'}(N_t-1)}{2^{B'}-1} \left(\frac{1}{2}\right)^{2^{B'}-1}
    \end{aligned}
\end{equation}
such that the term approaches 0 faster than exponentially. 

For the second part  of \eqref{equ1minuseta}, with the same choice of $t$, we have 
\begin{equation}
\begin{aligned}
        \mathbb{E}\Bigl[\left.\frac{1-\eta_{11}}{\eta_{11}}\right|&\eta_{11}>t\Bigr]Pr(\eta_{11}>t)\\
        &\le \mathbb{E}[\frac{1-\eta_{11}}{t}|\eta_{11}>t]Pr(\eta_{11}>t)\\
        & \le 2(N_t-1)\mathbb{E}[1-\eta_{11}]\\
        & \approx (N_t-1) 2^{-\frac{B'}{N_t-1}+1}.
\end{aligned}
\end{equation}
So, when $|\Delta \beta| = \frac{\epsilon\sigma^2}{P}\frac{1-\eta_{11}}{\eta_{11}H_1}$, it approaches 0 at least exponentially.

\noindent\textit{Case 2:}  When $|\Delta \beta| = \frac{\epsilon\sigma^2}{P}|\frac{(1-\eta_{11})H_1-\delta}{(\eta_{11}H_1+\delta)H_1}|$, 

If $\frac{(1-\eta_{11})H_1-\delta}{(\eta_{11}H_1+\delta)H_1}>0$, we know that $\frac{1-\eta_{11}}{\eta_{11}H_1}> \frac{(1-\eta_{11})H_1-\delta}{(\eta_{11}H_1+\delta)H_1}$ always holds, and maximization will not choose this term. As a result, for Case 2, we only need to consider the case where $\frac{(1-\eta_{11})H_1-\delta}{(\eta_{11}H_1+\delta)H_1}<0$:
\begin{equation}
    \begin{aligned}
        \left|\frac{(1-\eta_{11})H_1-\delta}{(\eta_{11}H_1+\delta)H_1}\right|<\left|\frac{\delta}{(\eta_{11}H_1+\delta)H_1}\right|<\left|\frac{\delta}{H_1^2}\right|
    \end{aligned}
\end{equation}
$H_1>\frac{\epsilon\sigma^2}{P}$ is a prerequisite of the NOMA's feasible region. So, we have
\begin{equation}
    \begin{aligned}
         \mathbb{E}\left|\frac{\delta}{H_1^2}\right| \le \frac{P^2}{\epsilon^2\sigma^4} \mathbb{\delta} \propto 2^{-B}
    \end{aligned}
\end{equation}
As a result, $|\Delta \beta_{min}|$ approaches 0 at least exponentially as $B$ and $B'$ increase.

\noindent\textbf{Step 2:} when $\beta^*=\beta_{\max}$

From definition, $\beta_{\max} = \min\{\frac{PH_{22}-\epsilon \sigma^2}{PH_{22}+\epsilon PH_{21}}, \frac{PH_{12}-\epsilon \sigma^2}{PH_{12}+\epsilon PH_{11}}\}$. We consider the two possible minimum cases separately. For each case, we express an upper bound of $\Delta \beta$ as a function of $\eta$ and $\delta$, and derive the convergence rate.

\noindent\textit{Case 1: $\beta_{\max} = \frac{PH_{22}-\epsilon \sigma^2}{PH_{22}+\epsilon PH_{21}}$}

\begin{equation}
\begin{aligned}
    |\Delta \beta |
    & = \left|\frac{P\hat{H}_{22}-\epsilon \sigma^2}{P\hat{H}_{22}+\epsilon P\hat{H}_{21}} - \frac{PH_{2}-\epsilon \sigma^2}{PH_{2}+\epsilon PH^*_{21}}\right|\\
    & = \left|\frac{\epsilon P \left(PH_{21}^*+\sigma^2\right)\left(H_2-\hat{H}_{22}\right)}{\left(P\hat{H}_{22}+\epsilon P\hat{H}_{21}\right)\left(PH_2+\epsilon P H_{21}^*\right)}\right.\\
    & \qquad+\left.\frac{\epsilon P \left(\epsilon \sigma^2 - P H_2\right)\left(H_{21}^*-\hat{H}_{21}\right)}{\left(P\hat{H}_{22}+\epsilon P\hat{H}_{21}\right)\left(PH_2+\epsilon P H_{21}^*\right)}\right|
\label{eq:28}
\end{aligned}
\end{equation}

\begin{lemma}
    $|H^*_{21}-\hat{H}_{21}|<2H_2(\sqrt{2-2\eta_{11}}+\sqrt{2-2\eta_{22}})+\delta+(1-\eta_{22})H_2$
    \label{lemmaS:2}
\end{lemma}
\noindent\textit{Proof.} 
\begin{equation}
\begin{aligned}
    &|H^*_{21}-\hat{H}_{21}| \\
    &= \left|H^*_{21}-H_{21}+H_{21}-\hat{H}_{21}\right|\\
    & \le \left|H^*_{21}-H_{21}\right|+\left|H_{21}-\hat{H}_{21}\right|\\
    & = \left|\left|\mathbf{h}_2^H\frac{\mathbf{h}_1}{\|\mathbf{h}_1\|}\right|^2-\left|\mathbf{h}_2^H\mathbf{w}_1\right|^2\right| +\left|\left|\mathbf{h}_2^H\mathbf{w}_1\right|^2-\hat{H}_{22}\left|\mathbf{w}_2^H\mathbf{w}_1\right|^2\right|\\
    & = \left|\left|\mathbf{h}_2^H\frac{\mathbf{h}_1}{\|\mathbf{h}_1\|}\right|^2-\left|\mathbf{h}_2^H\mathbf{w}_1\right|^2\right|\\
    &\quad +\left|\left|\mathbf{h}_2^H\mathbf{w}_1\right|^2-H_2\left|\mathbf{w}_2^H\mathbf{w}_1\right|^2\right|+\delta+(1-\eta_{22})H_2
\label{eq:29}
\end{aligned}
\end{equation}
Following a similar process used in the proof of  Lemma~\ref{lemma:S1}, we have
\begin{equation}
    \left|\left|\mathbf{h}_2^H\mathbf{w}_1\right|^2-H_2\left|\mathbf{w}_2^H\mathbf{w}_1\right|^2\right| < 2 H_2 \sqrt{2 - 2\eta_{22}}
\label{eq:30}
\end{equation}
Applying Eq.~\eqref{eq:30} to Eq.~\eqref{eq:29}, we complete the proof of Lemma~\ref{lemmaS:2}.

Applying Lemma~\ref{lemmaS:2} to Eq.~\eqref{eq:28}, we have
\begin{equation}
\begin{aligned}
    |\Delta \beta |
    & \le \frac{\epsilon P \left|PH_{21}^*+\sigma^2\right|\left((1-\eta_{22})H_2+\delta\right)}{P\eta_{22}H_2(PH_2+\epsilon P H_{21}^*)}\\
    & +\frac{\epsilon P \left|\epsilon \sigma^2 - P H_2\right|\left(2 H_2 \left(\sqrt{2-2\eta_{11}}+\sqrt{2-2\eta_{22}}\right)\right)}{P\eta_{22}H_2\left(PH_2+\epsilon P H_{21}^*\right)}\\
    &+\frac{\epsilon P \left|\epsilon \sigma^2 - P H_2\right|(\delta+(1-\eta_{22})H_2)}{P\eta_{22}H_2\left(PH_2+\epsilon P H_{21}^*\right)}
\end{aligned}
\end{equation}
We can rearrange the above terms as follows.
\begin{equation}
\begin{aligned}
    |\Delta \beta |
    & \le \frac{\left[\epsilon  \left|PH_{21}^*+\sigma^2\right|+\epsilon\left|\epsilon\sigma^2-PH_2\right|\right](1-\eta_{22})}{\eta_{22}\left(PH_2+\epsilon P H_{21}^*\right)}\\
    & +\frac{\left(\epsilon \left|PH_{21}^*+\sigma^2\right|+ \epsilon \left|\epsilon \sigma^2 - P H_2\right| \right)\delta}{\eta_{22}H_2\left(PH_2+\epsilon P H_{21}^*\right)}\\
    &+\frac{2 \epsilon \left|\epsilon \sigma^2 - P H_2\right|\left( H_2 \left(\sqrt{2-2\eta_{11}}+\sqrt{2-2\eta_{22}}\right)\right)}{\eta_{22}H_2\left(PH_2+\epsilon P H_{21}^*\right)}
\end{aligned}
\label{betamax}
\end{equation}
The upper bound above has three terms. In what follows, we prove that each term approaches 0 exponentially
as the rate increases. 
For the first term of Eq.~\eqref{betamax}, we have
\begin{equation}
\begin{aligned}
    &\frac{\left[\epsilon  \left|PH_{21}^*+\sigma^2\right|+\epsilon\left|\epsilon\sigma^2-PH_2\right|\right]\left(1-\eta_{22}\right)}{\eta_{22}\left(PH_2+\epsilon P H_{21}^*\right)} \\
    & \le \frac{\epsilon\left(2PH_2+(1+\epsilon)\sigma^2\right)\left(1-\eta_{22}\right)}{\eta_{22}PH_2}\\
    & \le \left(\frac{2\epsilon}{\eta_{22}}+\frac{\epsilon(1+\epsilon)\sigma^2}{\eta_{22}PH_2}\right)(1-\eta_{22})
\end{aligned}
\label{eq:34}
\end{equation}
From the feasibility condition, we have $\beta_{2,max} = \frac{PH_2-\epsilon\sigma^2}{PH_2+\epsilon P H^*_{21}}>0$ and $H_2 > \frac{\epsilon\sigma^2}{P}$, as a result:
\begin{equation}
    \frac{1}{H_2}<\frac{P}{\epsilon\sigma^2}
\end{equation}
Applying this to Eq.~\eqref{eq:34}, we have
\begin{equation}
\begin{aligned}
    &\frac{\left[\epsilon \left|PH_{21}^*+\sigma^2\right|+\epsilon\left|\epsilon\sigma^2-PH_2\right|\right](1-\eta_{22})}{\eta_{22}\left(PH_2+\epsilon P H_{21}^*\right)}  \\
    & \le \left(\frac{2\epsilon}{\eta_{22}}+\frac{\epsilon(1+\epsilon)\sigma^2}{\eta_{22}P} \frac{P}{\epsilon\sigma^2}\right)(1-\eta_{22})\\
    & =(1+3\epsilon)\frac{1-\eta_{22}}{\eta_{22}}
\end{aligned}
\end{equation}
In Eq.~\eqref{equ1minuseta} and the following equations, we have shown that $\mathbb{E}\left[\frac{1-\eta_{11}}{\eta_{11}}\right]$ approaches 0 exponentially when $B'$ increases. Following a similar procedure, $\mathbb{E}[\frac{1-\eta_{22}}{\eta_{22}}]$ approaches 0 exponentially. Therefore,
the first term approaches 0 exponentially as $B'$ increases.

For the second term of Eq.~\eqref{betamax}, we have
\begin{equation}
    \begin{aligned}
        &\frac{\left(\epsilon \left|PH_{21}^*+\sigma^2\right|+\epsilon \left|\epsilon \sigma^2 - P H_2\right| \right)\delta}{\eta_{22}H_2(PH_2+\epsilon P H_{21}^*)} \\
        &\le \frac{\epsilon \left(PH_{21}^*+\sigma^2+PH_2+\epsilon\sigma^2\right)\delta}{\eta_{22}P H_2^2}\\
        & \le \frac{\epsilon \left(2 P H_2+\sigma^2+\epsilon\sigma^2\right)\delta}{\eta_{22}P H_2^2}\\
        & = \left[\frac{2\epsilon}{\eta_{22}H_2}+\frac{\epsilon(1+\epsilon)\sigma^2}{P\eta_{22}H_2^2}\right]\delta \\
        &\le \left[\frac{2\epsilon}{\eta_{22}\left(\frac{\epsilon\sigma^2}{P}\right)}+\frac{\epsilon(1+\epsilon)\sigma^2}{P\eta_{22}\left(\frac{\epsilon\sigma^2}{P}\right)^2}\right]\delta\\
        & = \frac{P}{\sigma^2}\left(3+\frac{1}{\epsilon}\right)\delta \frac{1}{\eta_{22}}
    \end{aligned}
\end{equation}
Take the expectation on both sides,
\begin{equation}
    \begin{aligned}
        &\mathbb{E}\left[\frac{(\epsilon |PH_{21}^*+\sigma^2|+\epsilon |\epsilon \sigma^2 - P H_2| )\delta}{\eta_{22}H_2(PH_2+\epsilon P H_{21}^*)} \right] \\
        & \le \mathbb{E}\left[\frac{P}{\sigma^2}\left(3+\frac{1}{\epsilon}\right)\delta \frac{1}{\eta_{22}}\right]\\
        & = \frac{P}{\sigma^2}\left(3+\frac{1}{\epsilon}\right)\delta \mathbb{E}\left[\frac{1-\eta_{22}}{\eta_{22}}+1\right]
    \end{aligned}
\end{equation}
As $B'$ increases, $\mathbb{E}[\frac{1-\eta_{22}}{\eta_{22}}+1]$ is $O(1)$ and the whole term becomes $O(2^{-B})$, i.e., it decreases to 0 exponentially as $B$ increases.

For the third term of Eq.~\eqref{betamax}, we have
\begin{equation}
    \begin{aligned}
        &\frac{2 \epsilon \left|\epsilon \sigma^2 - P H_2\right| H_2 \left(\sqrt{2-2\eta_{11}}+\sqrt{2-2\eta_{22}}\right)}{\eta_{22}H_2\left(PH_2+\epsilon P H_{21}^*\right)} \\
        & \le \frac{2 \epsilon \left(\epsilon \sigma^2 + P H_2\right)\left(\sqrt{2-2\eta_{11}}+\sqrt{2-2\eta_{22}}\right)}{\eta_{22}\left(PH_2+\epsilon P H_{21}^*\right)}\\
        & \le \frac{2 \epsilon \left(\epsilon \sigma^2 + P H_2\right)\left(\sqrt{2-2\eta_{11}}+\sqrt{2-2\eta_{22}}\right)}{\eta_{22}P H_2}\\
        & = \left(\frac{2\epsilon^2 \sigma^2}{\eta_{22}PH_2}+\frac{2\epsilon}{\eta_{22}}\right)\left(\sqrt{2-2\eta_{11}}+\sqrt{2-2\eta_{22}}\right)\\
        & \le \left(\frac{2\epsilon^2 \sigma^2}{\eta_{22}P} \frac{P}{\epsilon \sigma^2}+\frac{2\epsilon}{\eta_{22}}\right)\left(\sqrt{2-2\eta_{11}}+\sqrt{2-2\eta_{22}}\right)\\
        & = 4\epsilon \frac{1}{\eta_{22}}\left(\sqrt{2-2\eta_{11}}+\sqrt{2-2\eta_{22}}\right)
    \end{aligned}
\end{equation}
We already know that $\mathbb{E}[\frac{1}{\eta_{22}}]$ is $O(1)$ as $B'$ increases. Since $\sqrt{2-2\eta_{11}}+\sqrt{2-2\eta_{22}}$ is $O(2^{\frac{B'}{2(N_T-1)}})$, the third term also approaches 0 exponentially.

\noindent\textit{Case 2: $\beta_{\max} = \frac{PH_{12}-\epsilon \sigma^2}{PH_{12}+\epsilon PH_{11}}$} 

The proof for Case 2 is very similar to that of Case 1 and is omitted for brevity.

\noindent\textbf{Step 3:} when $\beta^*=\beta_{0}$

From the definition of a stationary point, $\beta_0$ satisfies 
\begin{equation}
    a\beta^2+b\beta+c=0
\end{equation}
In the full CSI case, we have
\[
a^* = H_1H^*_{21}\left(H_2-H^*_{21}\right)
\]
\[
b^* = 2H_1 \sigma^2\left(H_2-H^*_{21}\right)
\]
\[
c^* = \sigma^4\left(H_2-H_1\right)+\sigma^2\left(H^*_{21}-H_1\right)H_2
\]
In the limited feedback system, we have a similar format, but with different parameter values:
\[
a_q = \hat{H}_{11}\hat{H}_{21}\left(\hat{H}_{22}-\hat{H}_{21}\right)
\]
\[
b_q = 2\hat{H}_{11}\sigma^2\left(\hat{H}_{22}-\hat{H}_{21}\right)
\]
\[
c_q = \sigma^4\left(\hat{H}_{22}-\hat{H}_{11}\right)+\sigma^2\left(\hat{H}_{21}-\hat{H}_{11}\right)\hat{H}_{22}
\]
If we define $f_q(\beta) = a_q \beta^2 +b_q \beta +c_q$ and $f^*(\beta) = a^* \beta^2 +b^* \beta +c^*$, 
\begin{equation}
    f^*(\beta_{0,q})-f^*(\beta_0^*)=f^*(\beta_{0,q})=f^{*'}(\xi)(\beta_{0,q}-\beta_0^*)
\end{equation}
where $\beta_{0,q}$ and $\beta_0^*$ are $\beta_0$ in the limited feedback and the full CSI cases, respectively, and $\xi$ is a value between $\beta_{0,q}$ and $\beta_0^*$,
\begin{equation}
    \begin{aligned}
        \left|\Delta \beta_0\right| & = \left|\frac{f^*(\beta_{0,q})}{f^{*'}(\xi)}\right|\\
        & = \left|\frac{f^*(\beta_{0,q})-f_q(\beta_{0,q})}{f^{*'}(\xi)}\right|\\
        & = \left|\frac{(a_q-a^*)\beta_{0,q}^{2} +(b_q-b^*)\beta_{0,q} +c_q-c^*}{f^{*'}(\xi)}\right|\\
        & = \left|\frac{\Delta a \beta_{0,q}^{2} +\Delta b\beta_{0,q} +\Delta c}{2a^*\xi+b^*}\right|\\
        &\le \left|\frac{\Delta a}{b^*}\right| + \left|\frac{\Delta b}{b^*}\right| + \left|\frac{\Delta c}{b^*}\right|
    \end{aligned}
\end{equation}

Below, we consider the above three terms separately.

\noindent\textit{Part 1:} $\frac{\Delta a}{b^*}$
\begin{equation}
\begin{aligned}
    |\Delta a| &= \left|a_q - a^*\right|\\
    &=\left|\hat{H}_{11}\hat{H}_{21}\left(\hat{H}_{22}-\hat{H}_{21}\right)-H_1H_{21}^*\left(H_2-H_{21}^*\right)\right|\\
    &=\left|\hat{H}_{21}\left(\hat{H}_{22}-\hat{H}_{21}\right)\left(\hat{H}_{11}-H_1\right)\right.\\
    &\quad\left.+H_1\left(\hat{H}_{22}-\hat{H}_{21}\right)\left(\hat{H}_{21}-H_{21}^*\right)\right.\\
    &\quad\left.+H_1 H_{21}^*\left[\left(\hat{H}_{22}-H_2\right)-\left(\hat{H}_{21}-H^*_{21}\right)\right]\right|\\
    &\le H_2^2 \left|\hat{H}_{11}-H_1\right|+H_1\hat{H}_{22}\left|\hat{H}_{21}-H_{21}^*\right|\\
    &\quad+H_1 H_2\left[\left|\hat{H}_{22}-H_2\right|+\left|\hat{H}_{21}-H^*_{21}\right|\right]\\
    & \le \left[(1-\eta_{11})H_1+\delta\right]H_2^2+H_1 H_2\left[(1-\eta_{22})H_2+\delta\right]\\
    & \quad +2 H_1 H_2\left[2H_2\left(\sqrt{2-2\eta_{11}}+\sqrt{2-2\eta_{22}}\right)\right]\\
    & \quad +2 H_1 H_2\left[(1-\eta_{22})H_2+\delta\right]
\end{aligned}  
\end{equation}
Here, we define $\cos\theta = \frac{\left|\mathbf{h}_2^H\mathbf{h}_1\right|}{\|\mathbf{h}_1\| \|\mathbf{h}_2\|}$. Then, $b^*= 2H_1 \sigma^2(H_2-H^*_{21})=2H_1H_2\sigma^2\sin^2\theta$. Since $\mathbf{h}_1$ and $\mathbf{h}_2$ are independent, we have $\sin^2\theta \sim Beta(N_t-1,1)$, and the expectation
\begin{equation}
    \begin{aligned}
        \mathbb{E}\left[\frac{1}{\sin^2\theta}\right] &= \int_0^1 \frac{1}{\sin^2\theta} (N_t-1) \left(\sin^2\theta\right)^{N_t-2} d \sin^2\theta\\
        & = \int_0^1 (N_t-1) \left(\sin^2\theta\right)^{N_t-3} d \sin^2\theta,
    \end{aligned}
\end{equation}
when $N_t>2$, $\mathbb{E}\left[1/\sin^2\theta\right]=(N_t-1)/(N_t-2)$ and when $N_t=2$, the expectation is infinite.
\begin{equation}
    \begin{aligned}
        \mathbb{E}|\frac{\Delta a}{b^*}|
        &\le \mathbb{E}\left[\frac{H_1 H_2^2(1-\eta_{11})+H_2^2\delta +3H_1H_2^2(1-\eta_{22})}{2H_1H_2\sigma^2\sin^2\theta} \right. \\
        &\quad \left. + \frac{3H_1H_2\delta +4H_1 H_2^2\sqrt{2-2\eta_{11}}}{2H_1H_2\sigma^2\sin^2\theta} \right.\\
        &\quad \left. + \frac{4H_1 H_2^2\sqrt{2-2\eta_{22}}}{2H_1H_2\sigma^2\sin^2\theta}\right]\\
        & =\frac{1}{2\sigma^2} \frac{N_t-1}{N_t-2}\left[N_t \mathbb{E}(1-\eta_{11})+\frac{N_t}{N_t-1}\mathbb{E}[\delta]\right.\\
        &\quad\left. +3N_t\mathbb{E}(1-\eta_{22})+3\delta\right.\\
        &\quad\left. +4N_t\mathbb{E}\sqrt{2-2\eta_{11}}+4N_t\mathbb{E}\sqrt{2-2\eta_{22}}\right]\\
        & = O(2^{-B})+O(2^{\frac{-B'}{2(N_t-1)}})
    \end{aligned}
    \label{eq:a1}
\end{equation}
So, $\mathbb{E}\left|\frac{\Delta a}{b^*}\right|$ approaches 0 exponentially when $N_t>2$.

\noindent\textit{Part 2:} $\frac{|\Delta b|}{b^*}$. For $\Delta b$, we have
\begin{equation}
    \begin{aligned}
        |\Delta b| &=\left|b_q-b^*\right|\\
        & =\left|2\sigma^2\hat{H}_{11}\left(\hat{H}_{22}-\hat{H}_{21}\right) -2\sigma^2 H_{1}\left(H_{2}-H^*_{21}\right)\right|\\
        & = \left|2\sigma^2\left[\hat{H}_{11}\left(\hat{H}_{22}-H_2+H_{21}^*-\hat{H}_{21}\right)\right.\right.\\
        & \quad\left.\left. +\left(\hat{H}_{11}-H_1\right)\left(H_2-H_{21}^*\right)\right]\right|\Bigr]
    \end{aligned}
\end{equation}
Since $\hat{H}_{11} \le H_1$, and by Lemma~\ref{lemmaS:2}, we have
\begin{equation}
    \begin{aligned}
        |\Delta b|
        & \le 2\sigma^2\left[H_1\left(\left|\hat{H}_{22}-H_2\right|+\left|\hat{H}_{21}-H^*_{21}\right|\right)\right.\\
        & \quad\left. +H_2\left|\hat{H}_{11}-H_1\right|\right]\\
        & \le 2\sigma^2\Bigl[2H_1H_2(1-\eta_{22})+2H_1\delta+H_1H_2(1-\eta_{11})\Bigr.\\
        & \quad\Bigl. +H_2\delta+2H_1H_2(\sqrt{2-2\eta_{11}}+\sqrt{2-2\eta_{22}})\Bigr]
    \end{aligned}
\end{equation}

As a result, we have
\begin{equation}
    \begin{aligned}
        \mathbb{E}\frac{|\Delta b|}{b^*}
        &\le
        \mathbb{E}\frac{1}{H_{1}H_{2}\sin^2\theta}\Bigl[2H_1H_2(1-\eta_{22})+2H_1\delta\Bigr.\\
        &\quad\left. +2H_1H_2\left(\sqrt{2-2\eta_{11}}+\sqrt{2-2\eta_{22}}\right)\right.\\
        &\quad\Bigl.+H_1H_2(1-\eta_{11})+H_2\delta\Bigr] \\
        &\le
        \mathbb{E}\left[\frac{1}{\sin^2\theta}\right]\mathbb{E}\Bigl[2(1-\eta_{22})+\frac{2}{H_2}\delta+(1-\eta_{11})\Bigr.\\
        &\quad\Bigl.+\frac{\delta}{H_1}+2\left(\sqrt{2-2\eta_{11}}+\sqrt{2-2\eta_{22}}\right)\Bigr].
    \end{aligned}
    \label{eq:b1}
\end{equation}
When $N_t>2$, we have $\mathbb{E}\frac{|\Delta b|}{b^*} =O\left(2^{-B}\right)+O\left(2^{-\frac{B'}{2(N_t-1)}}\right)$

\noindent\textit{Part 3:}$\frac{|\Delta c|}{b^*}$. For $\Delta c$, we have
\begin{equation}
    \begin{aligned}
        \Delta c &= c_q - c^*\\
        &= \sigma^4 \left[\left(\hat{H}_{22}-H_2\right)-\left(\hat{H}_{11}-H_{1}\right)\right]\\
        &\quad+\sigma^2\left[\left(\hat{H}_{21}-\hat{H}_{11}\right)\hat{H}_{22}-\left(H_{21}^*-H_1\right)H_2\right]\\
        &= \sigma^4 \left[\left(\hat{H}_{22}-H_2\right)-\left(\hat{H}_{11}-H_1\right)\right]\\
        &\quad+\sigma^2\left[\left(\hat{H}_{21}-\hat{H}_{11}\right)\left(\hat{H}_{22}-H_2\right)\right]\\
        &\quad+\sigma^2H_2\left[\left(\hat{H}_{21}-H_{21}^*\right)-\left(\hat{H}_{11}-H_{1}\right)\right]
    \end{aligned}
\end{equation}

which results in 
\begin{equation}
    \begin{aligned}
        |\Delta c| &\le \sigma^4 \left[\left|\hat{H}_{22}-H_2\right|+\left|\hat{H}_{11}-H_{1}\right|\right]\\
        &\quad+\sigma^2\left[\left|\hat{H}_{11}-\hat{H}_{21}\right|\left|\hat{H}_{22}-H_2\right|\right]\\
        &\quad+\sigma^2H_2\left(\left|\hat{H}_{11}-H_1\right|+\left|\hat{H}_{21}-H_{21}^*\right|\right)\\
        &\le \sigma^4 \left[\left|\hat{H}_{22}-H_2\right|+\left|\hat{H}_{11}-H_{1}\right|\right]\\
        &\quad+\sigma^2\left[\left(H_1+H_2\right)\left|\hat{H}_{22}-H_2\right|\right]\\
        &\quad+\sigma^2H_2\left(\left|\hat{H}_{11}-H_1\right|+\left|\hat{H}_{21}-H_{21}^*\right|\right)\\
        & \le\sigma^2\left(H_1+2 H_2+\sigma^2\right)\left[(1-\eta_{22})H_2+\delta\right]\\
        &\quad +\sigma^2\left(\sigma^2+H_2\right)\left[(1-\eta_{11})H_1+\delta\right]\\
        &\quad +2\sigma^2H_2^2 \left(\sqrt{2-2\eta_{11}}+\sqrt{2-2\eta_{22}}\right)
    \end{aligned}
\end{equation}

Dividing by $b^*$ and taking the expectation, when $N_t>2$, we get
\begin{equation}
    \begin{aligned}
        \mathbb{E}\frac{|\Delta c|}{b^*}
        &\le
        \mathbb{E}\Bigl[\frac{\left(H_1+2 H_2+\sigma^2\right)\left[(1-\eta_{22})H_2+\delta\right]}{2H_{1}H_{2}\sin^2\theta}\Bigr.\\
        &\Bigl.\quad + \frac{(H_2+\sigma^2)\left[(1-\eta_{11})H_1+\delta\right]}{2H_{1}H_{2}\sin^2\theta}\Bigr.\\
        &\quad\Bigl. + \frac{2 H_2^2\left(\sqrt{2-2\eta_{11}}+\sqrt{2-2\eta_{22}}\right)}{2H_{1}H_{2}\sin^2\theta}\Bigr]\\
        & \le \mathbb{E}\left\{ \left(\frac{2}{H_1}+\frac{1}{H_2}+\frac{\sigma^2}{H_1H_2}\right)[(1-\eta_{22})H_2+\delta]\right.\\
        &\quad \left.+\left(\frac{\sigma^2}{H_1 H_2}+\frac{1}{H_1}\right)\left((1-\eta_{11})H_1+\delta\right)\right.\\
        &\quad \left.+\left[\frac{2H_2}{H_1}(\sqrt{2-2\eta_{11}}+\sqrt{2-2\eta_{22}})\right]\right\}\mathbb{E}\frac{1}{2\sin^2\theta}\\
        &=O(2^{-B})+O(2^{-\frac{B'}{2(N_t-1)}}).
    \end{aligned}
    \label{eq:c1}
\end{equation}

As a result, we have proven that the $|\Delta\beta|$ decreases to 0 exponentially when $N_t>2$.

\noindent\textit{Part 4:} Extension to $N_t=2$.

Up to now, we have proven that the $|\Delta\beta|$ decreases to 0 exponentially when $N_t>2$, we now consider the case $N_t=2$. To complete the proof, we consider an auxiliary parameter $\tau$, which is independent from the specific channel, we complete the proof in two parts

\textit{Case 1:} $\sin^2\theta\geq\tau$

For $\Delta a$, following similar procedure in Eq.~\ref{eq:a1} and use the condition that $\sin^2\theta\geq\tau$, $N_t=2$, we have the equation,
\begin{equation}
\begin{aligned}
&\mathbb{E}\Big[
    \frac{|\Delta a|}{b^*}
    \mathbf{1}_{\{\sin^2\theta\geq\tau\}}
\Big]\\
&\leq
\frac{1}{2\sigma^2\tau}\left[N_t \mathbb{E}(1-\eta_{11})+\frac{N_t}{N_t-1}\mathbb{E}[\delta]\right.\\
        &\quad\left. +3N_t\mathbb{E}(1-\eta_{22})+3\delta\right.\\
        &\quad\left. +4N_t\mathbb{E}\sqrt{2-2\eta_{11}}+4N_t\mathbb{E}\sqrt{2-2\eta_{22}}\right]\\
&\leq
\frac{C_a}{\tau}
\Big(\delta+2^{-B^{\prime}/2}\Big),
\end{aligned}
\label{eq:nt2_delta_a}
\end{equation}
where $C_a$ is a constant independent of $\tau$, $B$ and $B'$.

For $\Delta b$, following Eq.~\ref{eq:b1} and use the condition that $\sin^2\theta\geq\tau$, we have the equation,
\begin{equation}
\begin{aligned}
&\mathbb{E}\Big[
    \frac{|\Delta b|}{b^{*}}
    \mathbf{1}_{\{\sin^2\theta\geq\tau\}}
\Big]\\
        &\le
        \frac{1}{\tau}\mathbb{E}\Bigl[2(1-\eta_{22})+\frac{2}{H_2}\delta+(1-\eta_{11})\Bigr.\\
        &\quad\Bigl.+\frac{\delta}{H_1}+2\left(\sqrt{2-2\eta_{11}}+\sqrt{2-2\eta_{22}}\right)\Bigr]\\
        &\leq
        \frac{C_b}{\tau}
        \Big(\delta+2^{-B^{\prime}/2}\Big).
\end{aligned}
\label{eq:nt2_delta_b}
\end{equation}
where $C_b$ is a constant independent of $\tau$, $B$ and $B'$.

For $\Delta c$, following Eq.~\ref{eq:c1} and use the condition that $\sin^2\theta\geq\tau$, we have the equation,
\begin{equation}
\begin{aligned}
&\mathbb{E}\Big[
    \frac{|\Delta c|}{c^{*}}
    \mathbf{1}_{\{\sin^2\theta\geq\tau\}}
\Big]\\
& \le \frac{1}{t}\mathbb{E}\left\{ \left(\frac{2}{H_1}+\frac{1}{H_2}+\frac{\sigma^2}{H_1H_2}\right)[(1-\eta_{22})H_2+\delta]\right.\\
        &\quad \left.+\left(\frac{\sigma^2}{H_1 H_2}+\frac{1}{H_1}\right)\left((1-\eta_{11})H_1+\delta\right)\right.\\
        &\quad \left.+\left[\frac{2H_2}{H_1}(\sqrt{2-2\eta_{11}}+\sqrt{2-2\eta_{22}})\right]\right\}\\
&\leq
    \frac{C_c}{2\tau}
    \Big(\delta+2^{-B^{\prime}/2}\Big).
\end{aligned}
\label{eq:nt2_delta_c}
\end{equation}
where $C_c$ is a constant independent of $\tau$, $B$ and $B'$.

\textit{Case 2:} $\sin^2\theta<\tau$, in this case, $\sin^2\theta \sim Beta(1,1)$, which is a uniform distribution, and since $|\Delta \beta_0|\le 1$, we have,

\begin{equation}
    \begin{aligned}
    &\mathbb{E}\Big[
        |\Delta\beta_0|
        \mathbf{1}_{\{\sin^2\theta<\tau\}}
    \Big]\\
    &\quad\leq \Pr(\sin^2\theta<\tau)\\
    &\quad =\tau.
    \end{aligned}
    \label{eq:nt2_small}
\end{equation}

As a result, we have,
\begin{equation}
    \begin{aligned}
    \mathbb{E}|\Delta\beta_0|
    &=
    \mathbb{E}\Big[
        |\Delta\beta_0|
        \mathbf{1}_{\{\sin^2\theta\geq\tau\}}
    \Big]\\
    &\quad+
    \mathbb{E}\Big[
        |\Delta\beta_0|
        \mathbf{1}_{\{\sin^2\theta<\tau\}}
    \Big]\\
    &\leq
    \frac{C_a+C_b+C_c}{\tau}
    \Big(\delta+2^{-B^{\prime}/2}\Big)+\tau.
    \end{aligned}
    \label{eq:nt2_combined}
\end{equation}

This inequation holds for any $B$ and $B'$ and any $0<\tau<1$. For sufficiently large $B$ and $B'$, we use $\tau = \sqrt{\delta+2^{-B'/2}}$, then Eq.~\ref{eq:nt2_combined} becomes,
\begin{equation}
    \begin{aligned}
    \mathbb{E}|\Delta\beta_0|
    &\leq
    (1+C_a+C_b+C_c)
    \sqrt{\delta+2^{-B^{\prime}/2}}\\
    &=
    O(2^{-B/2})+O(2^{-B^{\prime}/4}).
    \end{aligned}
\end{equation}
This proves that $\mathbb{E}\left[|\Delta\beta_0|\right]$ also converges to 0 at lease exponentially when $N_t=2$.

As a result, we have proven that the $|\Delta\beta|$ decreases to 0 exponentially for any $N_t\ge2$.

\bibliographystyle{IEEEtran}
\bibliography{IEEEabrv,sample}

\begin{thebibliography}{10}
\providecommand{\url}[1]{#1}
\csname url@samestyle\endcsname
\providecommand{\newblock}{\relax}
\providecommand{\bibinfo}[2]{#2}
\providecommand{\BIBentrySTDinterwordspacing}{\spaceskip=0pt\relax}
\providecommand{\BIBentryALTinterwordstretchfactor}{4}
\providecommand{\BIBentryALTinterwordspacing}{\spaceskip=\fontdimen2\font plus
\BIBentryALTinterwordstretchfactor\fontdimen3\font minus \fontdimen4\font\relax}
\providecommand{\BIBforeignlanguage}[2]{{%
\expandafter\ifx\csname l@#1\endcsname\relax
\typeout{** WARNING: IEEEtran.bst: No hyphenation pattern has been}%
\typeout{** loaded for the language `#1'. Using the pattern for}%
\typeout{** the default language instead.}%
\else
\language=\csname l@#1\endcsname
\fi
#2}}
\providecommand{\BIBdecl}{\relax}
\BIBdecl

\bibitem{10720669}
H.~Jafarkhani, H.~Maleki, and M.~Vaezi, ``Modulation and coding for {NOMA and RSMA},'' \emph{Proc. {IEEE}}, vol. 112, no.~9, pp. 1179--1213, 2024.

\bibitem{10786246}
W.~Chen \emph{et~al.}, ``Signal processing and learning for next generation multiple access in {6G},'' \emph{{IEEE} J. Sel. Topics Signal Process.}, vol.~18, no.~7, pp. 1146--1177, 2024.

\bibitem{Norouzi2023Joint}
S.~Norouzi, B.~Champagne, and Y.~Cai, ``Joint optimization framework for user clustering, downlink beamforming, and power allocation in {MIMO} {NOMA} systems,'' \emph{{IEEE} Trans. Commun.}, vol.~71, no.~1, pp. 214--229, Jan. 2023.

\bibitem{8823873}
M.~Vaezi \emph{et~al.}, ``Non-orthogonal multiple access: Common myths and critical questions,'' \emph{IEEE Wireless Commun.}, vol.~26, no.~5, pp. 174--180, 2019.

\bibitem{9832611}
A.~Mishra \emph{et~al.}, ``Rate-splitting multiple access for {6G}—part {I}: Principles, applications and future works,'' \emph{IEEE Commun. Lett.}, vol.~26, no.~10, pp. 2232--2236, 2022.

\bibitem{10729214}
Y.~Liu \emph{et~al.}, ``The road to next-generation multiple access: A 50-year tutorial review,'' \emph{Proc. IEEE}, vol. 112, no.~9, pp. 1100--1148, 2024.

\bibitem{3gpp_ts38214_v16.17}
3GPP, ``{NR; Physical layer procedures for data (Release 16)},'' {3GPP} TS 38.214, v16.17.0, Dec. 2023, available: https://www.3gpp.org.

\bibitem{7434594}
Z.~Ding and H.~V. Poor, ``Design of {Massive-MIMO-NOMA} with limited feedback,'' \emph{{IEEE} Signal Process. Lett.}, vol.~23, no.~5, pp. 629--633, 2016.

\bibitem{8502922}
J.~Cui, Z.~Ding, and P.~Fan, ``Outage probability constrained {MIMO-NOMA} designs under imperfect {CSI},'' \emph{{IEEE} Trans. Wireless Commun.}, vol.~17, no.~12, pp. 8239--8255, 2018.

\bibitem{9094017}
Y.~Yapıcı, I.~Guvenc, and H.~Dai, ``Low-resolution limited-feedback {NOMA} for {mmWave} communications,'' \emph{{IEEE} Trans. Wireless Commun.}, vol.~19, no.~8, pp. 5433--5446, 2020.

\bibitem{8935164}
M.~Morales-Céspedes, O.~A. Dobre, and A.~García-Armada, ``Semi-blind interference aligned {NOMA} for downlink {MU-MISO} systems,'' \emph{{IEEE} Trans. Commun.}, vol.~68, no.~3, pp. 1852--1865, 2020.

\bibitem{9281361}
H.~E. Hassani, A.~Savard, and E.~V. Belmega, ``Adaptive {NOMA} in time-varying wireless networks with no {CSIT/CDIT} relying on a 1-bit feedback,'' \emph{IEEE Wireless Commun. Lett.}, vol.~10, no.~4, pp. 750--754, 2021.

\bibitem{7968348}
X.~Liu and H.~Jafarkhani, ``Downlink non-orthogonal multiple access with limited feedback,'' \emph{{IEEE} Trans. Wireless Commun.}, vol.~16, no.~9, pp. 6151--6164, 2017.

\bibitem{9103094}
X.~Zou, M.~Ganji, and H.~Jafarkhani, ``Downlink asynchronous non-orthogonal multiple access with quantizer optimization,'' \emph{{IEEE} Wireless Commun. Lett.}, vol.~9, no.~10, pp. 1606--1610, 2020.

\bibitem{10384715}
M.~A. Almasi and H.~Jafarkhani, ``Rate loss analysis of reconfigurable intelligent surface-aided {NOMA} with limited feedback,'' \emph{IEEE Open J. Commun. Soc.}, vol.~5, pp. 856--871, 2024.

\bibitem{7972957}
Q.~Yang \emph{et~al.}, ``{NOMA} in downlink {SDMA} with limited feedback: Performance analysis and optimization,'' \emph{{IEEE} J. Sel. Areas Commun.}, vol.~35, no.~10, pp. 2281--2294, 2017.

\bibitem{4100151}
C.~K. Au-yeung and D.~J. Love, ``On the performance of random vector quantization limited feedback beamforming in a {MISO} system,'' \emph{IEEE Trans. Wireless Commun.}, vol.~6, no.~2, pp. 458--462, 2007.

\bibitem{1715541}
N.~Jindal, ``{MIMO} broadcast channels with finite-rate feedback,'' \emph{IEEE Trans. Inf. Theory}, vol.~52, no.~11, pp. 5045--5060, 2006.

\end{thebibliography}

\end{document}